\documentclass[nofootinbib,twocolumn,floatfix
]{revtex4-2}
\usepackage{xcolor}
\usepackage{subcaption}
\usepackage{amsmath,amsthm,amssymb}
\usepackage[normalem]{ulem}

\newcommand{\WpW}{\%\,\mathrm{W}^{-1}}
\usepackage{graphicx}
\usepackage{dcolumn}
\usepackage{bm}

\usepackage[justification=raggedright,
            singlelinecheck=false,
            font=small,
            labelfont=bf]{caption}

\begin{document}

\preprint{APS/123-QED}
\title{Sidewall-Poled Nanophotonic Lithium Niobate with Bidirectional Characterization}
\author{
Aditya Tripathi$^{1,\dagger}$,
Michael S. Bullock$^{1,4,\dagger}$,
Parash Thapalia$^{2, \dagger}$,
Pooja Kulkarni$^{2}$, 
Jaber Balalhabashi$^{6}$,
Robert Kwolek$^{1}$,
Shiva Behzadfar$^{2}$,
Kazuki Hirota$^{1}$,
Shion Eto$^{1}$, Rintaro Tsuchida$^{1}$,
Liam Beaudoin$^{1,7}$, 
Sasan Fathpour$^{2,3}$,
and Rajveer Nehra$^{1,4,5,6^\ast}$
}

\address{
$^{1}$Department of Electrical and Computer Engineering, University of Massachusetts Amherst, Amherst, MA 01003, USA\\
$^{2}$CREOL, The College of Optics and Photonics, University of Central Florida, Orlando, FL 32816, USA\\
$^{3}$Department of Electrical and Computer Engineering
University of Central Florida, Orlando, FL 32816, USA\\
$^{4}$College of Information and Computer Sciences, University of Massachusetts Amherst, Amherst, MA 01003, USA\\
$^{5}$Department of Physics, University of Massachusetts Amherst, Amherst, MA 01003, USA\\
$^{6}$Materials Science and Engineering Program, University of Massachusetts Amherst, Amherst, MA 01003, USA\\
$^{7}$Center for Photonic Communication and Computing, ECE Department, Northwestern University,
 Evanston, Illinois 60208, USA\\
 $^{\dagger}$These authors contributed equally to this work. \\
$^\ast$Email: rajveernehra@umass.edu\\
}

\date{\today} 

\begin{abstract}

Accurate characterization of integrated nonlinear photonic devices is often limited by unknown facet-coupling losses and fabrication-induced non-uniformities, leading to systematic over- or underestimation of the intrinsic on-chip performance. Here, we present and demonstrate a unified bidirectional characterization framework that exploits nonlinear interactions under forward and backward propagation to independently extract facet-specific coupling efficiencies, intrinsic nonlinear conversion efficiency, and the longitudinal quasi-phase-matching profile using only classical power measurements. We experimentally validate the method using sidewall-poled thin-film lithium niobate waveguides, obtaining a normalized second-harmonic generation efficiency of $(1850 \pm 20)~\%\mathrm{W}^{-1}$ while simultaneously demonstrating broadband non-degenerate optical parametric amplification and parametric generation spanning more than $10~\mathrm{THz}$.  Our framework is non-destructive, relies solely on classical power measurements, requires neither time-intensive microscopy nor calibrated internal references, and is compatible with wafer-scale testing, providing a general route to rigorous benchmarking and high-throughput characterization of photonic devices across material platforms.

\end{abstract}
\maketitle
\section{Introduction}\label{sec:level1}
\vspace{-3mm}
Integrated photonics offers a scalable pathway for both quantum and classical optical technologies by enabling the integration of photonic components, including light sources, reconfigurable interferometric circuits, optical delays, and photodetectors~\cite{chrostowski2015silicon}. Beyond miniaturization and high-density integration, photonics offers a key scalability advantage through spatially, frequency-, and time-multiplexed architectures without a proportional increase in physical footprint~\cite{obrien2009photonic, wang2020integrated}.  Silica and silicon nitride (SiN) photonics have enabled significant advances in large-scale optical systems, including programmable quantum circuits for generating, manipulating, and measuring nonclassical states of light~\cite{arrazola2021quantum, vaidya2020broadband, psiquantum2025manufacturable, yang2021squeezed}. These platforms benefit from mature CMOS-compatible foundry infrastructure, wafer-scale manufacturability, and ultralow propagation losses. However, their nonlinear functionality relies on intrinsically weak third-order ($\chi^{(3)}$) optical nonlinearities, which therefore require either high optical pump power or resonant structures.

In recent years, integrated platforms with second-order ($\chi^2$) optical nonlinearity, such as thin-film lithium niobate (TFLN), indium gallium phosphide (InGaP), and photogalvanic-induced ($\chi^2$) have rapidly emerged for frequency conversion, parametric amplification, and quantum light sources~\cite{nehra2022few, ledezma2022intense,jankowski2022quasi, hu2025efficient, li2025down, dean2026low}. As the nonlinear performance of these platforms continues to improve, the ability to accurately benchmark their intrinsic performance is equally important. Conventional characterization is typically performed using single-pass measurements, in which the measured conversion efficiency is inferred from the launched and collected optical powers. However, coupling losses often dominate the overall insertion loss and can vary between the input and output facets due to fabrication, polishing, and alignment errors. Consequently,  measured efficiencies are an inseparable combination of intrinsic conversion, propagation loss, and asymmetric coupling losses, leading to systematic biases in device performance. These uncertainties become increasingly significant when comparing nonlinear devices across a wafer, benchmarking fabrication processes, or accurately predicting the performance of quantum circuits, where even sub-decibel errors can lead to substantial inaccuracies in the inferred squeezing, entanglement generation, and quantum frequency-conversion efficiencies.

In particular, TFLN has rapidly emerged as a leading platform due to its strong optical nonlinearity, large electro-optical response,  and dispersion engineering capabilities~\cite{Wang2018, Zhang2017, vazimali2022applications,boes2023lithium,nehra2022few, jankowski2022quasi, ledezma2022intense, boes2018status}. In TFLN, efficient three-wave mixing has traditionally been achieved either through modal phase matching (MPM) or through quasi-phase matching (QPM). MPM exploits interactions between different spatial modes, but typically requires carefully engineered waveguide geometries and often sacrifices modal purity~\cite{boes2023lithium,arge2025demonstration}. Periodic poling, on the other hand, enables QPM while maintaining fundamental-mode operations, thereby accessing the largest nonlinear coefficient. However, a critical challenge in TFLN quantum photonics is poling-induced loss in widely employed \textit{poling-before-etching} approaches. In these devices, periodically inverted domains can exhibit differential etching, leading to sidewall corrugations and increased roughness~\cite{Rao:16}. Although these losses can be partially reduced by employing shallower etching and wider waveguides, such compromises reduce optical confinement and increase multimode operations~\cite{khalatpour2026roughness}. These limitations have motivated etching-before-poling strategies, such as sidewall poling, to achieve higher yields of QPM at desired wavelengths while enabling low-loss waveguides~\cite{Rao2019, Franken2026, sabatti2025nanodomain}. As these techniques continue to advance, it becomes critical to develop non-destructive, time-efficient in-situ optical measurements techniques to rigorously characterize the nonlinear performance of these devices. 

In this work, we introduce and experimentally demonstrate a bidirectional nonlinear characterization framework that leverages cascaded second-harmonic generation (SHG) and non-degenerate optical parametric amplification (NOPA) to unbiasedly determine the input and output coupling efficiencies and on-chip nonlinear performance of sidewall-poled TFLN devices.
By independently pumping the cascaded SHG-OPA interaction from opposite propagation directions, our approach decouples facet-dependent coupling losses from the intrinsic nonlinear response, enabling rigorous and unbiased extraction of on-chip nonlinear efficiencies. Beyond accurate efficiency measurements, the framework provides direct insight into the dominant mechanisms that limit device performance, including longitudinal thin-film non-uniformities and associated phase-matching variations, thereby establishing a practical methodology for the characterization and optimization of integrated nonlinear photonic circuits. Complementing this nonlinear characterization, we employ resonator-based measurements to directly compare conventional \textit{poling-before-etching} and \textit{sidewall-poling} fabrication strategies and quantify the excess loss introduced by periodic poling. For near-single-mode waveguides ($\sim1.3~\mu\mathrm{m}$ wide) at telecom wavelengths, we find that sidewall poling substantially suppresses poling-induced loss compared to the \textit{poling-before-etching} approach, while offering strong nonlinear interactions.

The remainder of this paper is organized as follows. Section~\ref{sec: Seciii_nonlinear} introduces the proposed bidirectional measurement framework for nonlinear characterization and discusses its application to extract unbiased nonlinear conversion efficiencies and phase-matching non-uniformities. We discuss the performance of our sidewall-poled devices in Section~\ref{subsec:gain}. Finally, Section~\ref{sec:conclusion} summarizes the key results and outlines future directions for low-loss nonlinear TFLN devices.

\section{Robust Bidirectional Nonlinear Characterization}\label{sec: Seciii_nonlinear}
Accurate benchmarking of integrated nonlinear photonic devices requires separating the intrinsic nonlinearity from the extrinsic optical losses of the measurement interface. Chip-to-fiber coupling losses often dominate total insertion loss and differ substantially between facets due to fabrication, polishing, and alignment asymmetries. Conventional single-pass measurements conflate these losses with propagation and nonlinear conversion efficiencies, biasing estimates of on-chip performance for quantum applications such as single-photon and squeezed-light generation. Meanwhile, characterizing longitudinal non-uniformities, which degrade quasi-phase matching, relies on time-intensive microscopy, SH imaging, or ellipsometric mapping that do not readily scale to wafer-level testing~\cite{zhao2023unveiling}. We address both challenges with a unified bidirectional characterization framework. Using only classical phase-insensitive power measurements, our framework independently extracts facet-specific coupling losses (Sec.~\ref{subsec:bd-description}), absolute SHG conversion efficiency (Sec.~\ref{subsec:SHG}), and longitudinal phase-matching uniformity (Sec.~\ref{subsec:poling-non-uniformity}), without requiring access to internal fields or destructive imaging. 

\subsection{Bidirectional Measurements for Unbiased Estimation of Coupling Efficiencies}\label{subsec:bd-description}
Conventional insertion-loss measurements provide access only to the total transmission efficiency $\eta\approx\eta_1\eta_2$, where $\eta_1$ and $\eta_2$ represent the input and output facet coupling efficiencies, respectively. Here, we introduce a bidirectional characterization approach that exploits a cascaded second-harmonic generation (SHG) and non-degenerate optical parametric amplification (NOPA) process to independently resolve facet-specific coupling efficiencies using only phase-insensitive classical power measurements, eliminating the need for loss- and noise-sensitive quantum measurements or above-threshold ring resonators \cite{wu26bnot,Chen2026universal}.

\begin{figure}[h]
    \centering\includegraphics[width=1\linewidth]{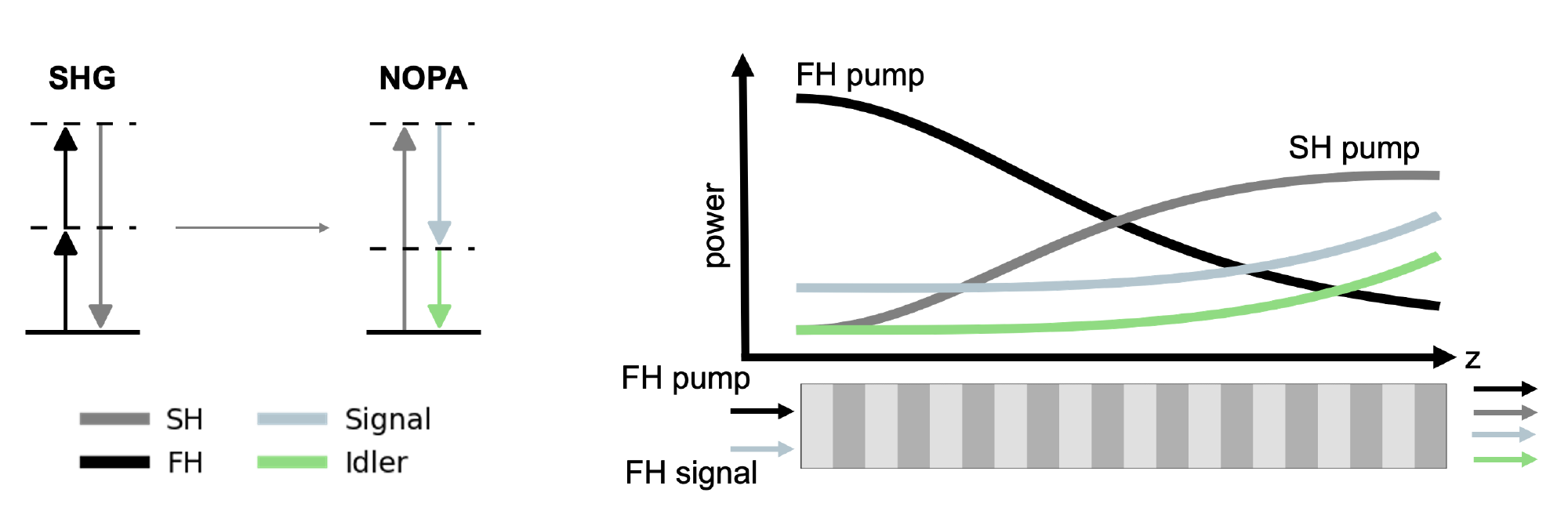}
    \caption{Schematic illustration of the cascaded SHG and NOPA process in a periodically poled lithium niobate waveguide. Left: Energy-level representation of the cascaded interaction. Right: Evolution of the optical powers along the propagation length. The FH pump is progressively depleted as SH power builds up. The generated SH field then drives parametric amplification, resulting in an amplified signal and generated idler fields.}
    \label{fig:nopa_scheme}
\end{figure}

The method employs a cascaded $\chi^{(2)}$ process detailed in Fig.~\ref{fig:nopa_scheme}: a fundamental harmonic (FH) pump at $\omega$ generates a second-harmonic pump field, which simultaneously drives NOPA of a co-propagating signal field at $\omega_s$, producing an idler at $\omega_i=2\omega-\omega_s$. Because the generated idler power depends nonlinearly on the coupled pump power, unequal facet efficiencies lead to markedly different detected idler mean photon numbers in the two propagation directions. In the low-gain regime, their measured power ratios simplify to 
\begin{equation}
\frac{n_{i,F}}{n_{i,B}}\approx \left(\frac{\eta_1}{\eta_2}\right)^2.
\end{equation}
Taking the square root, therefore, provides a direct closed-form estimate of the facet coupling asymmetry given by
\begin{equation}
\hat{R} = \sqrt{\frac{n_{i,F}}{n_{i,B}}} \approx \frac{\eta_1}{\eta_2}.
\label{eq:ratio_estimator}
\end{equation}
Combined with independently measured throughput $\eta_1\eta_2$, \eqref{eq:ratio_estimator} directly yields individual facet efficiencies. While closed-form and requiring no likelihood model or prior, the ratio estimate becomes biased for highly efficient cascaded interactions and strongly asymmetric facets. 

To overcome this, we use a maximum a posteriori (MAP) estimator that incorporates the prior on $\eta_1\eta_2$, defined as
\begin{align}
    (\hat\eta_1,\hat\eta_2) = \underset{\eta_1,\eta_2}{\mathrm{arg\,max}} \left[\max_{\kappa_{\mathrm{SHG}}}f(\mathbf{n}_{i}|\eta_1,\eta_2,\kappa_{\mathrm{SHG}})\,p(\eta_1\eta_2)\right],\label{eq:map}
\end{align} 
where $f(\mathbf{n}_{i}|\eta_1,\eta_2,\kappa_{\mathrm{SHG}})$ is the likelihood of the measured forward and backward idler nonlinear responses $\mathbf{n}_{i}$, $\kappa_{\mathrm{SHG}}$ is the effective nonlinear conversion parameter (treated as a nuisance parameter), and $p(\cdot)$ is the prior on total transmission $\eta_{\mathrm{tot}}=\eta_1\eta_2$ from propagation loss corrected linear transmission measurements. Full details and numerical studies are provided in the supplementary material~\ref{app:bd_npa}.

In the setup of Fig.~\ref{fig:bidir_nopa}a, a continuous-wave FH pump at 1536.75 nm is amplified using an erbium-doped fiber amplifier (EDFA), combined with a weak signal at 1537.50 nm through a 90:10 fiber coupler, and coupled into and out of the PPLN waveguide using single-mode lensed fibers. The idler generated is spectrally isolated by a 50~GHz narrowband filter and detected on a photodiode, while a fraction of the output is monitored using an optical spectrum analyzer (OSA). 

We apply this estimation method to a waveguide with a visible defect near the left facet. Fig.~\ref{fig:bidir_nopa}b shows $\sim$6~dB asymmetry between forward and backward-pumped idler powers. The MAP estimator yields $(\hat\eta_1,\hat\eta_2)$, attributing input/output facet losses of $8.9$ and $5.8$~dB and thus $3.1$~dB of excess loss at the defect. We verify the method by adding known loss via a variable optical attenuator (VOA) on the right facet, and plot the MAP and ratio estimator outputs against VOA-induced loss in Fig.~\ref{fig:bidir_nopa}c,d. The MAP estimator attributes the added loss to the right facet while the left facet estimate remains fixed, confirming that the two facets are independently characterized. The ratio estimator agrees when the facet losses are approximately equal, but acquires a slight bias when they are highly asymmetric. 

\begin{figure*}[!t]
    \centering
    \includegraphics[width=0.95\linewidth]{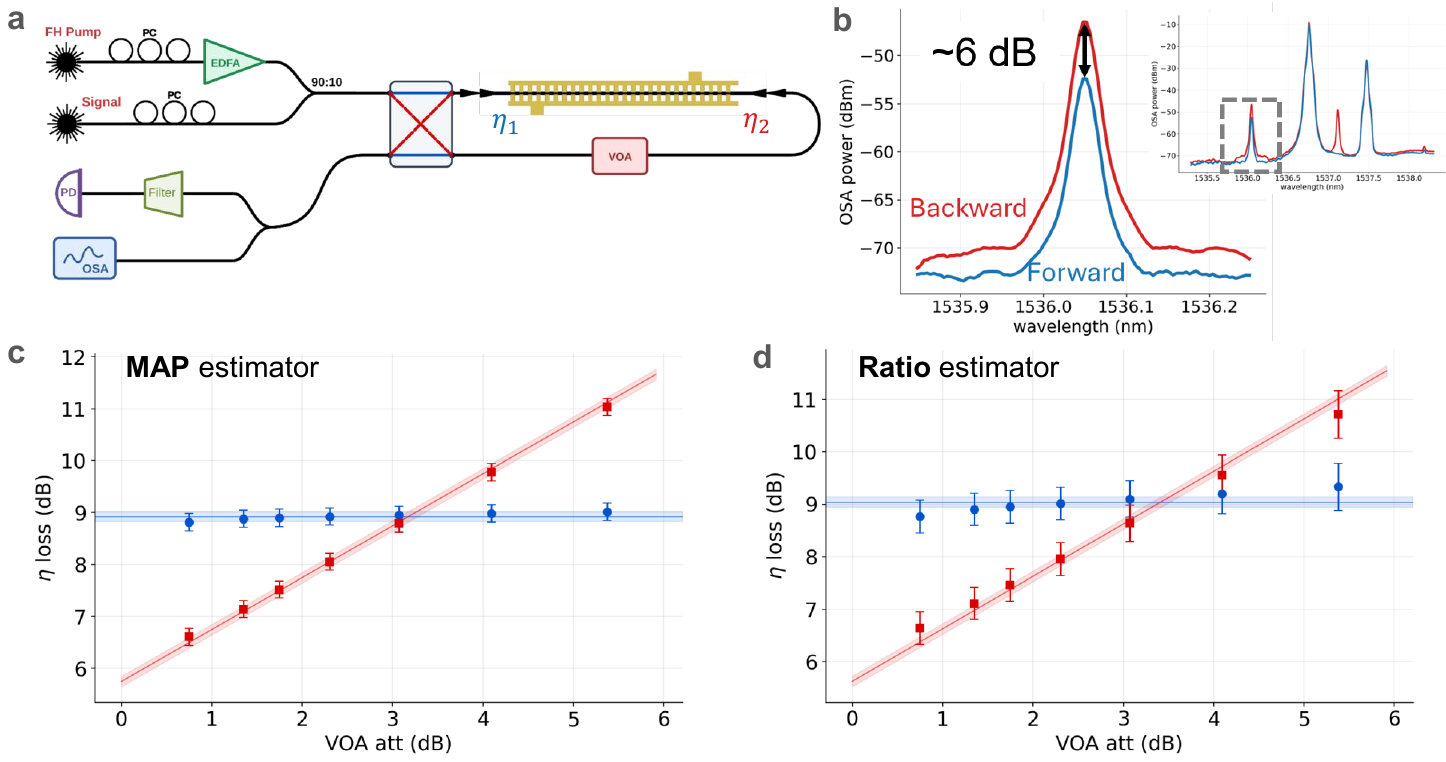}
    \caption{Bidirectional characterization of coupling losses. (a) Experimental setup for bidirectional facet-loss characterization. The fundamental-harmonic (FH) pump and signal are combined and launched into a PPLN waveguide, where cascaded SHG and NOPA generate the idler field. A switch reverses the propagation direction through the device, while a VOA is inserted at the right input facet to emulate controlled coupling loss. (b) Measured optical spectra of the generated idler for forward (blue) and backward (red) propagation. Inset: Full output spectrum showing the FH pump, signal, and generated idler. (c) Estimated facet coupling efficiencies, $\eta_1$ (blue) and $\eta_2$ (red), extracted using the MAP estimator as a function of the added VOA attenuation.  (d) Individual facet loss estimates of $\eta_1$ (blue) and $\eta_2$ (red) for the ratio estimator against VOA added attenuation. 
    }
    \label{fig:bidir_nopa}
\end{figure*}

\subsection{Second-Harmonic Generation Conversion Efficiency Characterization}\label{subsec:SHG}

Unbiased estimation of the intrinsic normalized SHG efficiency $\eta_{\mathrm{norm}}$ must account for coupling efficiencies that differ both between facets and between the FH and SH wavelengths. We denote by $\eta_{\omega,j}$ and $\eta_{2\omega,j}$ the FH and SH coupling efficiencies at facet $j\in\{0,1\}$, where we introduce the wavelength $\{\omega,2\omega\}$ index explicitly. A unidirectional measurement assuming symmetric facets, $\eta_{\omega,1}=\eta_{\omega,2}$ and $\eta_{2\omega,1}=\eta_{2\omega,2}$, yields a biased estimate of the on-chip efficiency. As in Sec.~\ref{subsec:bd-description} above, we remove this bias by combining forward and backward measurements.  

Assuming negligible propagation loss, a forward (resp. backward) pass couples the FH pump through facet 1 (resp. 2) and collects the SH through facet 2 (resp. 1). In the undepleted regime, fitting the detected SH power $P_{2\omega}$ against the square of the launched FH power $P_{\omega}$ yields the apparent direction-dependent conversion efficiencies
\begin{align}
    \eta_{\mathrm{app},F} &= \eta_{2\omega,2}\eta_\mathrm{norm}(\eta_{\omega,1})^2\\
    \eta_{\mathrm{app},B} &= \eta_{2\omega,1}\eta_\mathrm{norm}(\eta_{\omega,2})^2.
\end{align}
Each estimate is biased by their coupling factors, but their product depends on the coupling only through the total throughputs per-wavelength $\eta_{\omega}\equiv\eta_{\omega,1}\eta_{\omega,2}$ and $\eta_{2\omega}\equiv\eta_{2\omega,1}\eta_{2\omega,2}$, accessible from linear insertion-loss measurements. Normalizing the geometric mean by these throughputs cancels the facet asymmetry, 
\begin{align}
    \sqrt{\frac{\eta_{\mathrm{app},F}\eta_{\mathrm{app},B}}{\eta_\omega^2\eta_{2\omega}}} &= \eta_\mathrm{norm}. \label{eq:eta_norm_bd}
\end{align}
The result is independent of how the total coupling loss is partitioned between the two facets, so no assumption of facet symmetry is required~\cite{karnik202618dbonchipvacuumsqueezing}.

Using the setup of Fig.~\ref{fig:bidir_qpm}a, the forward and backward fits give biased apparent efficiencies $\eta_{\mathrm{app},F}=134~\WpW$ and $\eta_{\mathrm{app},B}=217~\WpW$, whose facet-asymmetry-corrected geometric mean yields an unbiased normalized conversion efficiency of $\eta_{\mathrm{norm}}=170~\WpW$. This estimator operates in the undepleted regime and requires the throughputs at both SH and FH, which can be difficult to obtain for devices with integrated routing such as wavelength division multiplexers (WDMs). 
When the device can instead be driven to pump depletion, $\eta_{\mathrm{norm}}$ is obtained from a single propagation direction using only the input FH facet loss, itself estimated using the MAP estimator of Section~\ref{subsec:bd-description}. 
In the lossless-propagation limit (supplementary material ~\ref{app:bd_npa}, Eq.~\eqref{eq:pump-depletion}) the on-chip SH power is 
\begin{align}
    P_{2\omega} &= P_\omega^{\mathrm{chip}}\tanh^{2}\left(\sqrt{\eta_{\mathrm{norm}}\,P_\omega^{\mathrm{chip}}}\right),
    \label{eq:sh-depletion}
\end{align}
where $P_\omega^{\mathrm{chip}}=\eta_\omega^{\mathrm{in}}P_\omega$ is the on-chip FH power set by the input-facet coupling $\eta_\omega^{\mathrm{in}}$ alone. 
The low-power quadratic slope is consistent with any value of the SH-arm coupling and cannot fix the absolute scale. 
The approach to saturation in \eqref{eq:sh-depletion}, by contrast, is pinned by energy conservation to $P_\omega^{\mathrm{chip}}$, so fitting the full $\tanh^2(.)$ curve determines $\eta_{\mathrm{norm}}$ absolutely from a single direction, without an independent SH-arm calibration or a backward measurement. Although \eqref{eq:sh-depletion} assumes lossless propagation, the FH- and SH-loss contributions bias the fitted $\eta_{\mathrm{norm}}$ in opposite directions and largely cancel. Loss corrections are more important for the estimation of absolute on-chip conversion efficiency $P_{2\omega}/P_{\omega}^{\rm chip}$~\cite{Franken2026}. For the high-efficiency $2.1~\mathrm{cm}$ device of Sec.~\ref{subsec:gain}, we therefore numerically integrate using loss-included coupled mode equations.

\begin{figure*}[!tbh]
    \centering
    \includegraphics[width=0.95\linewidth]{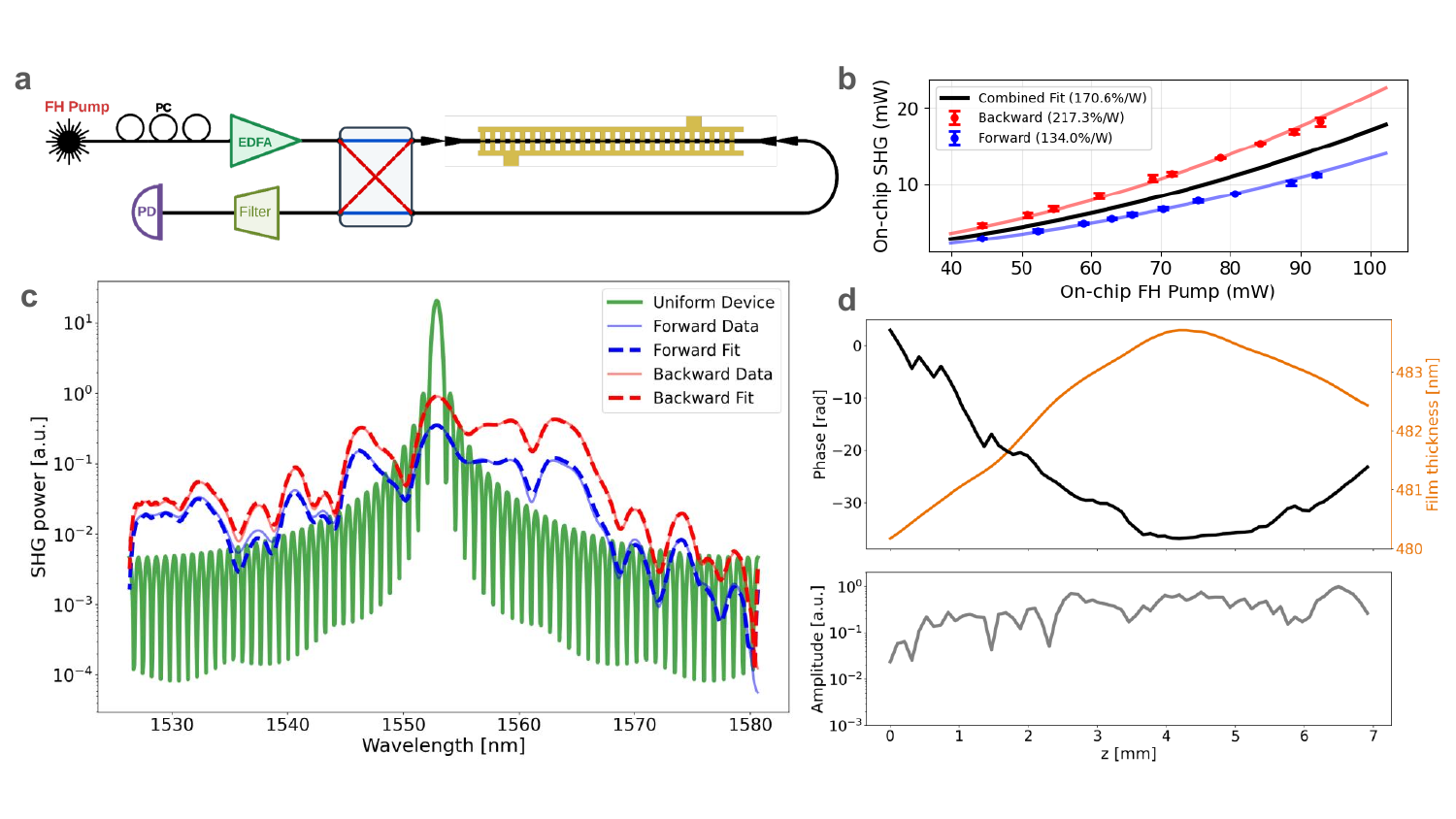}
    \caption{Bi-directional characterization of the SHG conversion efficiency and longitudinal quasi-phase matching nonuniformity. (a) Measurement setup for characterization. The FH pump is amplified and sent to a switch that controls the propagation direction, and the output generated by the SH is filtered and measured with a power meter. (b) Bi-directional SHG conversion efficiency characterization, where unbiased  SHG conversion efficiency is estimated as $170~\WpW$. (c) Longitudinal quasi-phase matching non-uniformity fits for forward (blue) and backward (red) propagation directions, along with simulated phase mismatch corrected QPM curve (green). (d) Upper: estimated longitudinal phase-mismatch function $h(z)$ (black), along with the measured ellipsometry data showing thin-film thickness variation for the device (orange). Lower: amplitude of estimated $h(z)$. }
    \label{fig:bidir_qpm}
\end{figure*} 

\subsection{Longitudinal Phasematching Non-uniformity Characterization}\label{subsec:poling-non-uniformity}
Beyond extracting the absolute SHG efficiency, the bidirectional measurement framework provides access to the longitudinal structure of the quasi-phase-matching grating. Imperfections such as duty-cycle errors, domain-wall position deviations, and phase mismatch from thin film thickness variation manifest as a complex-valued envelope $h(z)$ modulating the ideal periodic poling through the nonlinear coupling coefficient $d(z)$ as
\begin{align}
    d(z) = h(z) e^{i K_g z} + c.c. \label{eq:dz-maintext}
\end{align}
with $K_g = \frac{2\pi}{\Lambda}$ and where $h(z)$ is a square-integrable,
complex-valued envelope that varies slowly compared to the poling period $\Lambda$. 
Recovering $h(z)$ from intensity-only measurements constitutes a phase-retrieval problem. Here, we show that bidirectional SHG spectra supply sufficient information to solve it uniquely. In the undepleted-pump regime, the forward and backward SHG efficiencies are proportional to the squared magnitudes of the $z$-transform of $h(z)$ evaluated on two concentric circles in the complex plane (supplementary material~\ref{sec:appendix-uniqueness}, Eqs.~\eqref{eq:etaF} and~\eqref{eq:etaB}):
\begin{align}
    \eta_F(\Delta k) &\propto \left|H(r_1 e^{-i\phi})\right|^2, \label{eq:etaF-zt}\\
    \eta_B(\Delta k) &\propto \left|H(r_2 e^{-i\phi})\right|^2, \label{eq:etaB-zt}
\end{align}
with radii $r_1 = e^{\alpha \Delta z}$ and $r_2 = e^{-\alpha \Delta z}$ set by the effective loss asymmetry $\alpha = 2\alpha_\omega - \alpha_{2\omega}$, where $\alpha_{\omega}$ and $\alpha_{2\omega}$ are the propagation loss coefficients at FH and SH, respectively, as in \eqref{eq:alpha-eff}, and where $\phi \equiv \Delta k\,\Delta z$ is the normalized phase mismatch (supplementary material~\ref{sec:appendix-uniqueness}). As we prove in the supplementary material~\ref{sec:appendix-uniqueness} (Lemma~\ref{lem:main}), knowledge of $|H|$ on two circles of distinct radii determines $H$ uniquely up to a global phase, eliminating the reflection ambiguity inherent to single-measurement phase retrieval (Lemma~\ref{lem:single-circle}). The spatial resolution of the reconstruction is $\Delta z = 2\pi / \Delta k_{\mathrm{span}}$, where $\Delta k_{\mathrm{span}}$ is the experimentally accessible range of phase mismatch, tuned here via the pump wavelength.

We apply the inversion to a PPLN waveguide \textit{without adaptive poling} using the setup of Fig.~\ref{fig:bidir_qpm}a.  Sweeping the FH pump wavelength in each direction yields the forward and backward SHG efficiency spectra shown in Fig.~\ref{fig:bidir_qpm}c (blue and red traces). 

Discretizing the waveguide into $N$ segments, we solve the two-circle phase-retrieval problem numerically for the samples $h[n] = h(n\,\Delta z)$. The fitted spectra, together with a simulated phase-corrected $h(z)=|h(z)|$ curve (green), are shown in Fig.~\ref{fig:bidir_qpm}c, and the reconstructed envelope in Fig.~\ref{fig:bidir_qpm}d. Its recovered phase (upper panel, black) exhibits a slowly varying spatial drift that matches ellipsometry measurements of the thin-film thickness variation (orange), while the recovered amplitude (lower panel) remains relatively uniform, indicating that duty-cycle errors are modest compared with the accumulated phase drifts.\\

This observation motivates the two-pronged strategy employed for the high-efficiency devices reported in Sec.~\ref{subsec:gain}. First, wafer-scale dicing (supplementary material~\ref{app:tf-char}) minimizes the film thickness RMSE along the propagation direction. Second, post-etching adaptive poling adjusts $\Lambda(z)$ to cancel residual phase mismatch (supplementary material~\ref{app: Fabrication and Poling}). The bidirectional reconstruction demonstrated here closes the loop on this strategy: by non-invasively recovering the longitudinal phase profile from classical power measurements alone, it provides a direct, wafer-compatible diagnostic of the residual non-uniformity left after dicing and adaptive poling, and a means to verify that the remaining drift is small enough to support the high absolute conversion efficiencies and broadband parametric gain reported in the following section.

\section{Low-loss, Efficient Nonlinear Devices}
\label{subsec:gain}
\begin{figure*}
    \centering
    \includegraphics[width=0.95\linewidth]{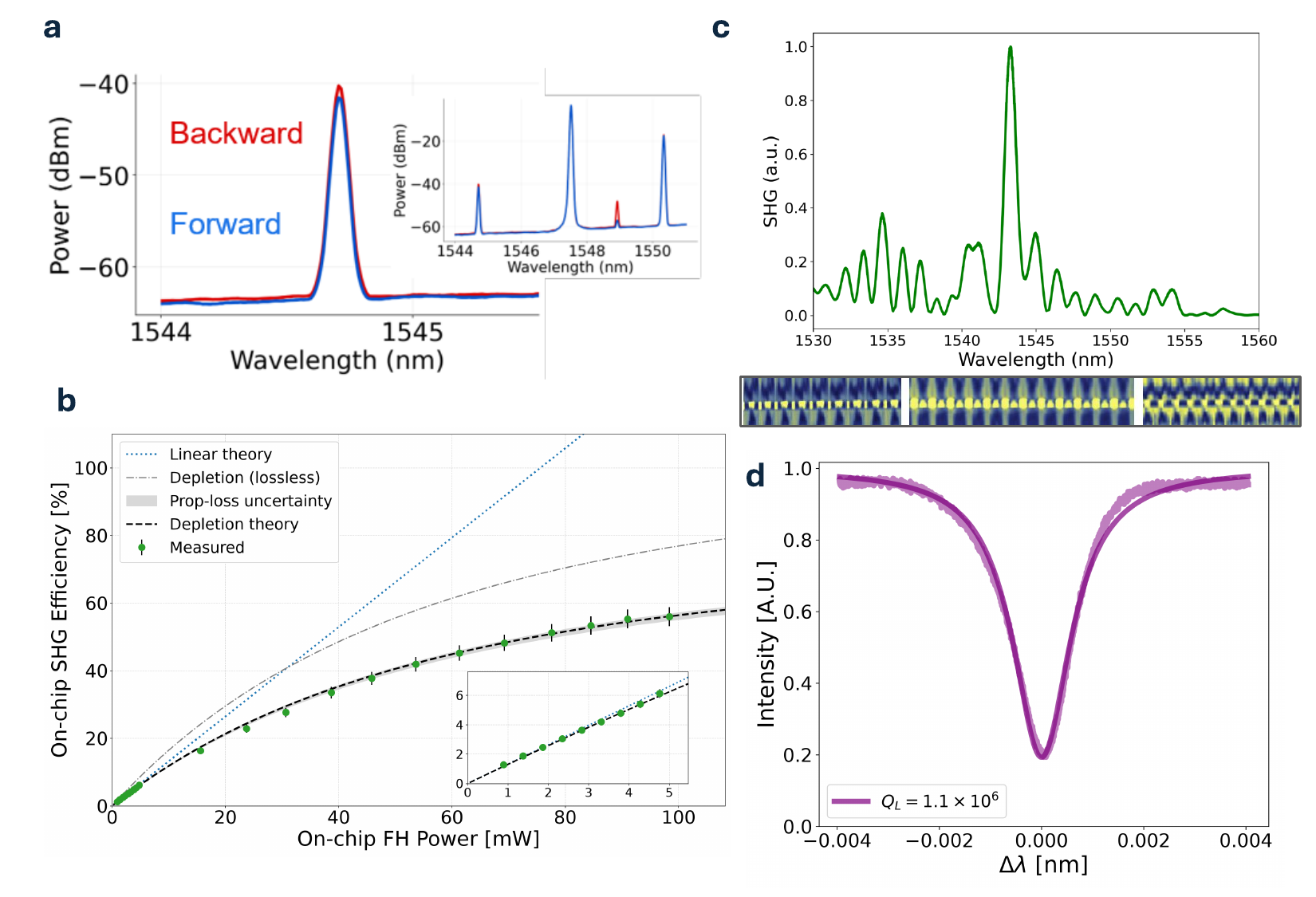}
    \caption{Characterization of the nonlinear and linear performance of sidewall-poled waveguides. (a) Bidirectional cascaded SHG--NOPA measurements used to determine the facet-specific coupling losses. Forward (blue) and backward (red) optical spectra exhibit different idler powers arising from asymmetric chip-to-fiber coupling.
(b) Absolute on-chip SHG conversion efficiency as a function of on-chip FH pump power. The solid line is a fit to the pump-depletion model, with 0.25 dB/cm FH propagation loss and 0.9 dB/cm SH propagation loss. \emph{Inset:} Low-power linear regime.
(c) Top: Measured SHG QPM spectrum, exhibiting a narrow phase-matching response indicative of good longitudinal poling uniformity. Bottom: Representative confocal microscopy images acquired at different positions along the waveguide, revealing the periodically poled domain structure together with local duty-cycle variations, lateral domain broadening, and incomplete domain penetration.
(d) Measured resonator response fabricated on the same chip together with the fitted response.}
    \label{fig:CE}
\end{figure*}

We now characterize the fabricated adaptively poled TFLN waveguides. A detailed discussion on the fabrication process and poling is provided in the supplementary material (supplementary material~\ref{app: fab_pole}). 

We characterize linear propagation loss using resonator measurements, while evaluating the nonlinear response through SHG, NOPA, and broadband optical parametric generation (OPG). These measurements collectively demonstrate the high nonlinear efficiency and ultra-broadband phase-matching enabled by sidewall-poled waveguides, with gain bandwidths exceeding those achievable in large mode areas LN devices~\cite{Kashiwazaki2020}.

We first apply the bidirectional MAP estimator of Sec.~\ref{subsec:bd-description} to determine the FH input coupling efficiency. The total FH insertion loss of the device is 9.90~dB, and from the asymmetry between the forward and backward cascaded SHG--NOPA idler power measurements in Fig.~\ref{fig:CE}a we extract an input-facet loss $\eta_{\mathrm{in},\omega}$ of $4.51$\,dB.

Figure~\ref{fig:CE}b shows the measured continuous-wave SHG conversion efficiency of the 2.1-cm-long adaptively poled waveguide. Using calibrated facet coupling efficiencies, SH power is referenced to the on-chip fundamental power and fitted to the numerically integrated loss-included pump-depletion model discussed after Eq.~\eqref{eq:sh-depletion}. 
From the fit, we extract an intrinsic normalized SHG efficiency of
$\eta_{\mathrm{norm}} = 1850 \pm 20~\%\,\mathrm{W}^{-1}$, corresponding to an estimated SH output-facet loss of $9.28\pm0.1~\mathrm{dB}$. The device achieves an absolute on-chip conversion efficiency of approximately $56\%$ at an on-chip fundamental power of $95~\mathrm{mW}$, among the highest reported for sidewall-poled TFLN waveguides operating between 1550 and 775~nm~\cite{Franken2026,sabatti2025nanodomain}. 

The inset shows the expected quadratic scaling $P_{2\omega}\propto P_{\omega}^2$ at low power, with saturation at higher power marking the onset of pump depletion. We emphasize that without bidirectional loss correction, the same data would yield a $\eta_{\mathrm{norm}}$ biased by input-facet loss, underscoring the importance of the proposed bidirectional framework for rigorous absolute benchmarking of integrated nonlinear devices.

Figure~\ref{fig:CE}c shows the SHG QPM spectrum together with confocal microscopy of the poled waveguide. A single, narrow phase-matching peak indicates good longitudinal uniformity of the nonlinear grating over the 2.1-cm nonlinear interaction length. The confocal images reveal a well-defined domain-inversion pattern (middle) alongside local duty-cycle variations (left), lateral domain broadening (right), and partial poling depth, effects also observed in \textit{poling-before-etching} approaches~\cite{Bollmers2025}. We attribute these imperfections to spatial variations in photoresist thickness, which render the coercive field nonuniform for poling. As a result, the effective nonlinearity is below the TFLN platform's intrinsic capability. Nevertheless, the absence of appreciable spectral broadening or additional phase-matching peaks indicates the robustness and reproducibility of the sidewall-poling fabrication process, demonstrated across multiple devices on multiple chips. 

Finally, we characterize the propagation loss of our devices through resonator measurements and obtain extrinsic and intrinsic quality factors of $(Q_e = (3.75 \pm 0.34)\times10^6,\;Q_i =(1.41 \pm 0.14)\times10^6)$ obtained from 6 resonator responses, corresponding to the resonator and coupler. We conservatively attribute the larger of the two to the extrinsic (coupling) and the smaller to the intrinsic quality factor respectively. 
Since the resonator waveguide's top width of 
$1.2~\mu\mathrm{m}$ is narrower than the nonlinear device's ($1.5~\mu\mathrm{m}$), the mode overlap with the etched sidewalls is greater. Thus, the propagation loss of the 
nonlinear device is expected to be lower than this bound. More details on resonator measurements are provided in the supplementary material (supplementary material~\ref{app:sidewall-added-loss}). 

\begin{figure*}[!htb]
    \centering
    \includegraphics[width=0.98\linewidth]{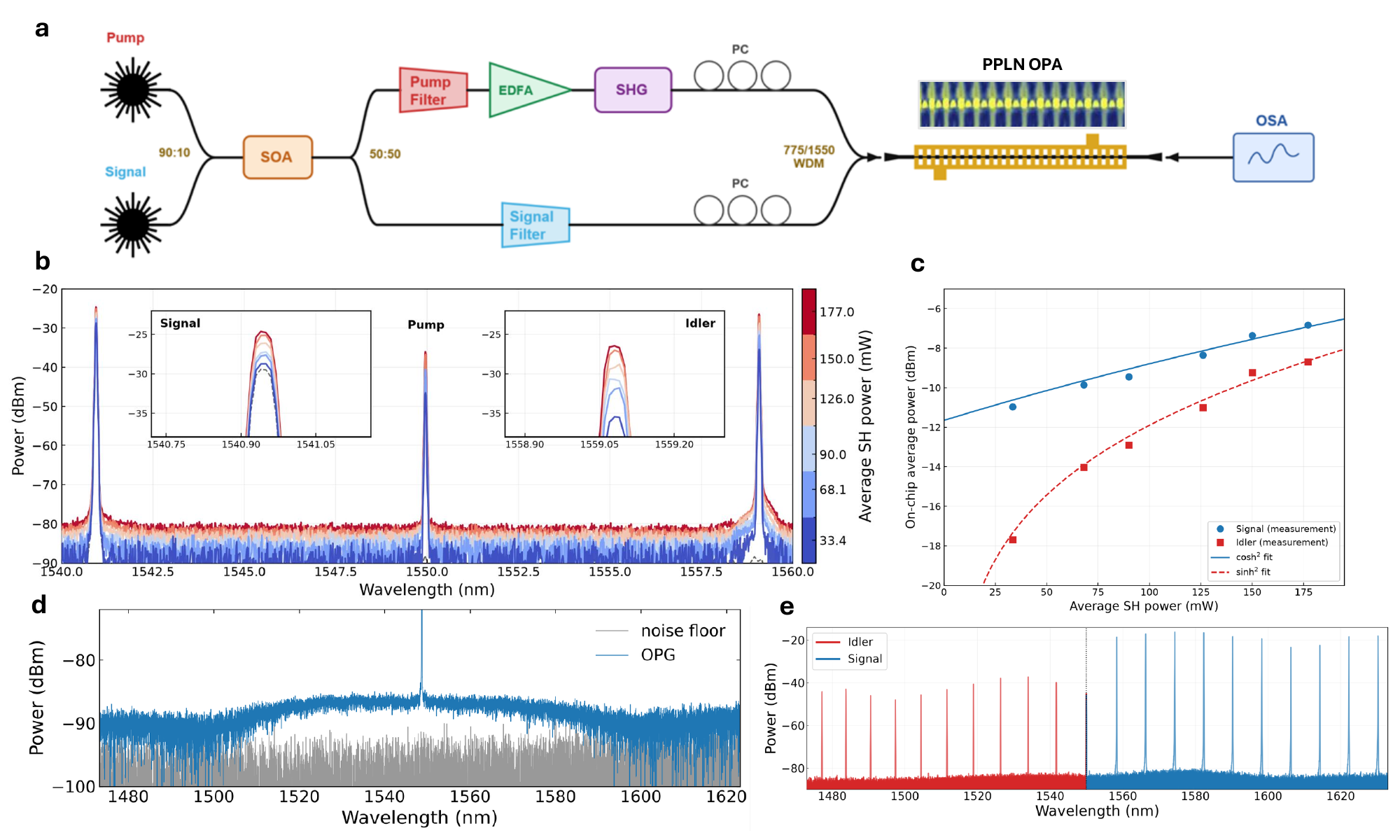}
    \caption{Broadband non-degenerate optical parametric amplification in sidewall-poled TFLN waveguides. (a) Experimental setup. The FH pump and signal are combined and modulated into nanosecond pulses. The two fields are subsequently separated, and the FH pump is amplified and frequency-doubled to generate the SH pump. The SH pump and signal are recombined using a fiber wavelength-division multiplexer (WDM), coupled into the chip, and the output spectrum is analyzed with an optical spectrum analyzer (OSA). (b) Measured output spectra for average off-chip SH pump powers ranging from 33~mW to 177~mW at a fixed signal wavelength of approximately 1541~nm. The apparent pump peak arises from a second-order artifact of the OSA generated by the 775~nm SH pump. (c) Extracted on-chip signal and idler powers as a function of the average off-chip SH pump power. (d) Optical parametric generation (OPG) spectrum acquired with the input signal removed, demonstrating spontaneous broadband signal--idler generation. (e) Continuous-wave NOPA spectra obtained by tuning the input signal over a 75~nm wavelength range, demonstrating broadband parametric amplification. }
    \label{fig:gainvpumppower}
\end{figure*}

Having established intrinsic nonlinear performance, we next evaluate the parametric gain performance via phase-insensitive gain and optical parametric generation (OPG). In Fig.~\ref{fig:gainvpumppower}(a), a tunable laser (Toptica CTL) provides the signal and a fixed-wavelength laser (NKT Basik) seeds the FH pump. The two are combined using a 90:10 coupler and modulated into 10-ns pulses at a repetition rate of 5.73\,MHz with a high-speed optical switch (BOA1550PXS). Pulsed operation provides high peak power at low average power, minimizing thermal drifts, photorefractive effects, and associated optical damage. The pump arm is amplified using an erbium-doped fiber amplifier and a frequency-doubled SHG module (NKT Harmonik) to produce the pump at the SH wavelength. The SH pump is then recombined with the signal off-chip and coupled into the device through a lensed fiber. At the output, the amplified signal, generated idler, and residual pump are collected with a second lensed fiber and analyzed using an optical spectrum analyzer (OSA).

 Figure~\ref{fig:gainvpumppower}b shows the measured output spectra for increasing SH pump powers. The signal inset with dashed trace correspond to the transmitted signal without the SH pump. At the maximum available off-chip average SH pump  of 177 mW, we measure a phase-insensitive signal gain of 4.8~dB, corresponding to an estimated phase-sensitive gain of 9.98~dB. The estimated on-chip signal and idler powers versus off-chip average SH pump power (Fig.~\ref{fig:gainvpumppower}(c)) show the signal-idler difference decreasing to a minimum of 1.85~dB at the highest pump power, consistent with the high-gain regime where the signal and idler photon numbers become nearly equal (see eq. \eqref{eq:n_idler_int}, \eqref{eq:n_signal_int} in the supplementary material). 

Finally, we characterize the broadband parametric response. With the input signal removed, the OPG spectrum of Fig.~\ref{fig:gainvpumppower}d, generated solely by spontaneous parametric down-conversion, provides a seed-independent verification of the phase-matching bandwidth, with a 3-dB bandwidth exceeding 10~THz. This broad parametric response is enabled by dispersion engineering, which exhibits a low group-velocity dispersion (GVD) of only $0.16$--$0.22~\mathrm{ps}^2/\mathrm{m}$ across the 1510---1630~nm wavelength range and approximately $0.30~\mathrm{ps}^2/\mathrm{m}$ at the SH pump wavelength of 775~nm. 
Figure~\ref{fig:gainvpumppower}e shows the NOPA measurement with the CW signal tuned from 1558~nm to 1630~nm at a fixed average SH power of 177~mW, demonstrating the device's near-flat gain.

\section{Conclusion and Outlook}
\label{sec:conclusion}
We have introduced a bidirectional measurement framework for their rigorous nonlinear characterization for integrated devices and demonstrate it using low-loss, high-efficiency sidewall-poled nanophotonic lithium niobate waveguides. Our bidirectional framework uses only phase-insensitive classical power measurements to deconvolve facet-dependent coupling losses from the intrinsic nonlinear response. By pumping a cascaded SHG--NOPA interaction from opposite propagation directions and combining a closed-form ratio estimator with a maximum a posteriori estimator constrained by an independent insertion-loss prior, we resolved individual facet coupling efficiencies that are robust to measurement noise and accurate even under strong facet asymmetry, as confirmed by numerical simulation and a controlled VOA verification. Applying this loss-deconvolution to absolute SHG benchmarking, we extracted an unbiased normalized efficiency of $\eta_{\mathrm{norm}} = 1850 \pm 20~\%\,\mathrm{W}^{-1}$, and estimate $56\%$ absolute on-chip SH conversion. Furthermore, we show how our framework further enabled a non-invasive reconstruction of the longitudinal quasi-phase-matching profile: we showed that bidirectional SHG spectra reduce, after discretization, to a two-circle phase-retrieval problem with a provably unique solution, and that the recovered phase of the nonlinear envelope $h(z)$ tracks the independently measured thin-film thickness variation. 

More broadly, the phase-mismatch-resolved response $\eta(\Delta k)$ encodes the modal dispersion of the interacting waves, giving access to further device parameters. For instance, the effective-index difference $n_{2\omega}-n_\omega=\lambda_{\mathrm{pm}}/2\Lambda$ can be determined from the phase-matched wavelength, and the group-velocity mismatch from the SHG acceptance bandwidth. Although this work focuses on SHG, a bidirectional sum-frequency variant would decouple the two input frequencies to map group-velocity dispersion across the full frequency plane, with the facet-loss deconvolution and phase-retrieval machinery carrying over unchanged. Finally, we demonstrate the operational performance of these devices as broadband parametric amplifiers, measuring a phase-insensitive gain of $4.8~\mathrm{dB}$ (corresponding to a phase-sensitive gain of $\sim10~\mathrm{dB}$) and amplification maintained across a signal tuning range $75~\mathrm{nm}$, with seed-off operation that confirms broadband optical parametric generation over 10 THz bandwidth. 

On the fabrication side, we directly compare \textit{poling-before-etching} and \textit{etching-before-poling} (sidewall-poled) processes using a resonator-embedded loss-extraction technique that isolates the poling-induced contribution from background propagation and coupling losses. We find that poling prior to waveguide definition introduces substantial excess propagation loss compared to sidewall poling—which we attribute to domain-dependent differential etching and the resulting increase in sidewall roughness. We use conformal oxide cladding to improve the poling depth and the uniformity of the duty cycle, which we confirm to be $42.5\%\pm1.5\%$ for the sidewall-poled waveguides. We attribute the residual departure from an ideal duty cycle $50\%$ in sidewall poling
to a proximity-driven lateral spread, in which the fringing
fields of adjacent inverted domains overlap most strongly in the interior,
where the domains reach and exceed the target duty cycle and merge before the
fringing field-limited ends fully invert~\cite{Nakamura2002,Krasnokutska2021}.  On the other hand, poled-before-etched waveguides show $55.3\%\pm5.8\%$ with reduced poling uniformity.

Although we partially mitigate this through multiple rounds of poling and confocal microscopy to optimize the length of the device with maximum depth and minimal lateral spread, this can be mitigated further by pre-compensating the electrode fill factor well below $50\%$~\cite{Krasnokutska2021}, segmenting the poled regions into independently optimized poling sections~\cite{Bollmers2025}, terminating the poling pulse via in-situ SHG monitoring~\cite{Rao2019,Niu2020}, or applying bipolar preconditioning pulse sequences that reduce the internal field and suppress lateral domain spread~\cite{NagyReano2020}. 

Taken together, the low-loss sidewall-poling process and the bidirectional metrology framework address two central bottlenecks in scaling TFLN photonics: the trade-off between optical confinement and poling-induced loss, and the conflation of facet coupling with intrinsic device performance in conventional single-pass measurements. The methods reported here are non-destructive, require only classical power detection, and are directly compatible with wafer-level testing, making them well suited to the high-throughput characterization demanded by large-scale photonic circuits in TFLN and beyond. We expect these advances to facilitate the scalable development of nonlinear TFLN devices in applications requiring high modal purity and low loss, including efficient quantum frequency conversion and integrated light sources for quantum information processing and communications.

\section{Acknowledgments}
This work was supported in part by the DARPA INSPIRED Program (D24AC00154-00). The views, opinions, and/or findings expressed are those of the author(s) and should not be interpreted as representing the official views or policies of the U.S. Government or any agency thereof. RN acknowledges the support of the College of Engineering, University of Massachusetts Amherst. The UCF contribution work is partially supported by NSF Industry University Cooperative Research Center (IUCRC) EPICA program. Use of AI tools: The authors used Claude Opus 4.8 (Anthropic) to assist with manuscript preparation, including editing for clarity, LaTeX formatting, and checking notational consistency, as well as to assist in the verification of derivations and the development of numerical simulation and analysis code. All scientific content, analysis, and conclusions were developed and validated by the authors, who take full responsibility for the content of this work.

\bibliography{Final}
\clearpage
\onecolumngrid

\section*{Supplementary Material}

\appendix

\section{Uniqueness of QPM Inversion and Thin-Film Variation Characterization}\label{sec:appendix-uniqueness}
\newtheorem{theorem}{Theorem}
\newtheorem{lemma}[theorem]{Lemma}
\newtheorem{corollary}[theorem]{Corollary}
\theoremstyle{remark}
\newtheorem*{remark}{Remark}
\newcommand{\CC}{\mathbb{C}}
\newcommand{\RR}{\mathbb{R}}
\newcommand{\re}{\operatorname{Re}}

In this section, we show that recovering the longitudinal nonlinear-coupling non-uniformity from bidirectional SHG intensity measurements in a lossy waveguide reduces, after discretization, to a two-circle phase retrieval problem. Consider a PPLN waveguide of length~$L$ with wavevectors $k_\omega$, $k_{2\omega}$, amplitude propagation loss coefficients $\alpha_\omega$, $\alpha_{2\omega}$, target poling period $\Lambda$, and longitudinal nonlinear coupling coefficient
\begin{align}
    d(z) = h(z) e^{i K_g z} + c.c. \label{eq:dz}
\end{align}
with $K_g = \frac{2\pi}{\Lambda}$ and where $h(z)$ is a square-integrable,
complex-valued envelope that varies slowly compared to $\Lambda$ and
absorbs both amplitude deviations (duty-cycle errors) and phase
deviations (domain-wall position errors) from ideal first-order poling. We neglect higher grating harmonics ($m=\pm 3,\pm 5,\dots$) under
the assumption $|\Delta k|\ll 2K_g$.

\noindent{\textbf{Forward propagation.}}
The pump enters at $z = 0$ and decays as
$A_\omega(z) = A_0\,e^{-\alpha_\omega z}$.
The forward SH envelope $A_{2\omega}^{(F)}$, driven by
the nonlinear polarization
$P_{\text{NL}} = d(z)\,A_\omega^2(z)\,e^{i2k_\omega z}$, obeys
\begin{equation}\label{eq:cwe-fwd}
  \frac{dA^{(F)}_{2\omega}}{dz}
  = -\alpha_{2\omega}\,A^{(F)}_{2\omega}
      + i\,A_0^2\;
        h(z)\,e^{-2\alpha_\omega z}\,e^{i\Delta k\,z}.
\end{equation}
Multiplying \eqref{eq:cwe-fwd} by the integrating factor
$e^{\alpha_{2\omega}z}$ and integrating from~$0$ to~$L$ with boundary
condition $A^{(F)}_{2\omega}(0) = 0$ gives
\begin{equation}\label{eq:fwd-sol}
  A^{(F)}_{2\omega}(L)
    = iA_0^2\,e^{-\alpha_{2\omega}L}
      \int_0^L h(z)\,
              e^{-(2\alpha_\omega-\alpha_{2\omega})z}\,
              e^{i\Delta k\,z}\,dz.
\end{equation}
Define the effective loss coefficient
\begin{equation}\label{eq:alpha-eff}
  \alpha \;\equiv\; 2\alpha_\omega - \alpha_{2\omega}.
\end{equation}
The phase mismatch is 
\begin{equation}
    \Delta k=k_{2\omega}-2k_\omega-\frac{2\pi}{\Lambda}.
\end{equation}
Since the prefactor $e^{-\alpha_{2\omega}L}$ is independent of
$\Delta k$, the SHG efficiency is
\begin{equation}\label{eq:etaF}
  \eta_F(\Delta k)
\propto\left|
    \int_0^L h(z)\,e^{-\alpha z}\,e^{i\Delta k\,z}\,dz\right|^2.
\end{equation}

\noindent{\textbf{Backward propagation.}}
For backward pumping, the pump enters at $z=L$ with decay $e^{-2\alpha_\omega(L-z)}$, and the counter-propagating fields select the $h^*(z)e^{-iK_g z}$ component of \eqref{eq:dz} as the driving term. Repeating the integrating-factor argument, the resulting integral is the pointwise conjugate of that of $h(z)\,e^{\alpha z}\,e^{i\Delta k\,z}$, and therefore has equal modulus, i.e.,
\begin{align}\label{eq:etaB}
  \eta_B(\Delta k)
\propto\left|
    \int_0^L h(z)\,e^{\alpha z}\,e^{i\Delta k\,z}\,dz\right|^2.
\end{align}

The forward and backward SHG intensities are thus the squared magnitudes of Fourier-type integrals of $h(z)$ weighted by
$e^{\mp\alpha z}$. Phase retrieval from a single magnitude measurement is ambiguous: the zeros of any solution may be independently reflected across the measurement circle, a result known for the general class of Paley-Wiener functions \cite{akutowicz,walther,hofstetter}. Lemma~\ref{lem:single-circle} gives a self-contained proof for the discrete case, and Lemma~\ref{lem:main} shows that imposing magnitude agreement on two distinct circles forces all zeros to coincide, eliminating the ambiguity.
\begin{lemma}[Single-circle ambiguity]\label{lem:single-circle} Let $F(z)=f[0]\prod_{k=1}^{N-1}(1-z_k z^{-1})$. If $|G(re^{i\omega})|=|F(re^{i\omega})|$ for all $\omega$, then $G$ is obtained by independently replacing each zero $z_k$ with either $z_k$ or $r^2/\bar z_k$, up to a global phase.
\end{lemma}
\begin{proof}
The identity $|F(re^{i\omega})|^2 = F(z)\cdot\overline{F[0]}\prod_k(1-\bar z_k z r^{-2})\big|_{z=re^{i\omega}}$ shows that the Laurent polynomial $P_r(z)=F(z)\cdot F^*(r^2/z)$ has roots $(z_k, r^2/\bar z_k)$, organized into $N{-}1$ conjugate-reciprocal pairs. Equal magnitudes on the circle forces $G$ to have the same $P_r$, so $G$ must select one root from each pair.
\end{proof}
\begin{lemma}[Two-circle uniqueness]\label{lem:main} Let $H(z)=\sum_{n=0}^{N-1}h[n]\,z^{-n}$. If $|\tilde H|=|H|$ on circles of radii $r_1\neq r_2>0$, then $\tilde H =\lambda H$ for some $\lambda \in \mathbb{C}$.
\end{lemma}
\begin{proof}
Suppose some zero $z_a$ of $H$ is not a zero of $\tilde H$. By Lemma \ref{lem:single-circle} on circle 1, $\tilde H$ has zero $r_1^2/\bar z_a$ in its place. By Lemma \ref{lem:single-circle} on circle 2, that same new zero equals $r_2^2/\bar z_b$ for some zero $z_b$ of $H$. If $b=a$: then $r_1^2/\bar z_a = r_2^2/\bar z_a$, so $r_1=r_2$, which is a contradiction.
If $b\neq a$: then $|z_b|=(r_2^2/r_1^2)|z_a|$. Note that $z_b$ is also displaced, so repeating the argument gives a chain of distinct zeros with $|z^{(n)}|=(r_2/r_1)^{2n}|z_a|\to 0$. Since a finite-degree polynomial cannot have infinitely many zeros, we have a contradiction. Thus, all zeros agree, giving $\tilde H=\lambda H$.
\end{proof}

We partition $[0,L]$ into $N$ segments of width $\Delta z=L/N$ and define the samples $h[n]=h(n\Delta z)$. Approximating the integrals in \eqref{eq:etaF} and \eqref{eq:etaB} by rectangle sums and introducing the normalized frequency $\phi=\Delta k \Delta z$ gives \begin{align}
    \eta_F(\phi)&\propto \left|\sum_{n=0}^{N-1}h[n]e^{-\alpha n \Delta z} e^{i\phi n}\right|^2\\
        \eta_B(\phi)&\propto \left|\sum_{n=0}^{N-1}h[n]e^{+\alpha n \Delta z} e^{i\phi n}\right|^2.
\end{align}
Define the z-transform $H(z)=\sum_{n=0}^{N-1}h[n]z^{-n}$. We find that 
\begin{align}
    \eta_F(\phi)&\propto \left|H(r_1e^{-i\phi
    })\right|^2\\
    \eta_B(\phi)&\propto \left|H(r_2e^{-i\phi
    })\right|^2
\end{align}
with radii $r_1=e^{\alpha \Delta z}$ and $r_2=e^{-\alpha \Delta z}=1/r_1$: the forward and backward measurements evaluate $|H|$ on reciprocal circles outside and inside the unit circle. 
The radii are distinct whenever $\alpha = 2\alpha_\omega - \alpha_{2\omega}\neq0$, which is generically satisfied since absorption and scattering at the second harmonic differ from twice the fundamental value. When $\alpha$ is small enough that measurement noise degrades conditioning, oversampling and averaging the SH data improves accuracy. As $\phi$ sweeps from $0$ to $2\pi$, the physical phase mismatch traverses $\Delta k_\mathrm{span}=2\pi/\Delta z$, so the spatial resolution is $\Delta z = 2\pi/\Delta k_\mathrm{span}$, where $\Delta k_\mathrm{span}$ is the range of phase mismatch accessible by varying the pump frequency and the PPLN temperature, and the number of resolvable segments is $N=L/\Delta z = L\Delta k_\mathrm{span}/2\pi$. We must measure at $2N-1$ equally spaced normalized frequencies per circle, since $|H(re^{-i\phi})|^2$ is a trigonometric polynomial of degree $N{-}1$ in $\phi$, spaced uniformly in $\Delta k$ based on simulation of $\Delta k(\omega,T)$ for our waveguide geometries.\\

\noindent\textbf{Numerical recovery of $h[n]$.}
We recover the complex-valued samples $h[n]$ by nonlinear least-squares minimization. Writing the $z$-transform matrices $[\mathbf{M}_F]_{m,n} = r_1^{-n}e^{i\phi_m n}$ and $[\mathbf{M}_B]_{m,n} = r_2^{-n}e^{i\phi_m n}$, the model predictions are $|\mathbf{M}_{F}\mathbf{h}|^2$ and $|\mathbf{M}_{B}\mathbf{h}|^2$. Because the proportionality constants relating $\eta_{F}$ and $\eta_{B}$ to $|H|^2$ are unknown, we introduce scale factors $A_F$ and $A_B$ that are determined analytically at each iteration by projecting the data onto the current model
\begin{align}
  A_F = \frac{\langle\eta_F,|\mathbf{M}_F\mathbf{h}|^2\rangle}{\||\mathbf{M}_F\mathbf{h}|^2\|^2},
\end{align}
and likewise for $A_B$. To improve conditioning, the residuals are formulated in amplitude rather than intensity as
\begin{equation}
  e_{F,m} = \sqrt{A_F}\,|[\mathbf{M}_F\mathbf{h}]_m| - \sqrt{\eta_F(\phi_m)}, \qquad
  e_{B,m} = \sqrt{A_B}\,|[\mathbf{M}_B\mathbf{h}]_m| - \sqrt{\eta_B(\phi_m)}.
\end{equation}
A smoothness penalty is appended to the residual vector using second-order finite differences of the amplitude profile,
\begin{equation}
  e_{\mathrm{reg},n} = \mu\bigl(|h[n{-}1]| - 2|h[n]| + |h[n{+}1]|\bigr),
\end{equation}
with regularization weight $\mu = 0.1$. The $2N$ real decision variables (real and imaginary parts of $\mathbf{h}$) are optimized using SciPy's \texttt{least\_squares}, starting from a uniform real-valued initial guess $h[n]=1$.

\section{Bidirectional Cascaded SHG-OPA for Facet-Loss Estimation}\label{app:bd_npa}
\subsection{Cascaded SHG-OPA forward model}\label{app:forward_model}

We work with photon-flux amplitudes $A_j(z)$ such that $|A_j(z)|^2 = n_j(z)$ is the photon number per mode at frequency $j \in \{\omega,\,2\omega,\,s,\,i\}$ at position $z \in [0,L]$. The near-degenerate pump, signal, and idler share the per-length amplitude propagation loss $\alpha_\omega$, and the second harmonic decays at $\alpha_{2\omega}$. Denoting the SHG and OPA per-length nonlinear couplings $\kappa_{\mathrm{SHG}}$ and $\kappa_{\mathrm{OPA}}$, the slowly-varying envelope equations for the cascaded process are
\begin{equation}
    \begin{aligned}
        \frac{d A_\omega}{d z}      &= -\kappa_{\mathrm{SHG}}\,A_\omega A_{2\omega} - \alpha_\omega\,A_\omega, \\
        \frac{d A_{2\omega}}{d z}   &= \tfrac{1}{2}\,\kappa_{\mathrm{SHG}}\,A_\omega^{2} -\alpha_{2\omega}\,A_{2\omega}, \\
        \frac{d A_s}{d z}           &= +\kappa_{\mathrm{OPA}}\,A_{2\omega} A_i-\alpha_\omega A_s, \\
        \frac{d A_i}{d z}           &= +\kappa_{\mathrm{OPA}}\,A_{2\omega} A_s-\alpha_\omega A_i.
    \end{aligned}
    \label{eq:fields}
\end{equation}

Because the SHG and OPA stages share the same waveguide, the same effective nonlinearity $d_{\mathrm{eff}}$, and near-degenerate signal/idler wavelengths, we set $\kappa_{\mathrm{OPA}}=\kappa_{\mathrm{SHG}}$. The boundary conditions on the input facet are $A_\omega(0) = \sqrt{\eta_{\mathrm{in}}\,n_{\omega,\mathrm{src}}}$, $A_s(0) = \sqrt{\eta_{\mathrm{in}}\,n_{s,\mathrm{src}}}$, and $A_{2\omega}(0) = A_i(0) = 0$, where $n_{\omega,\mathrm{src}}$ and $n_{s,\mathrm{src}}$ are the source photon numbers per pulse delivered to the input fiber and $\eta_{\mathrm{in}}$ is the input-facet coupling.

Under the experimental conditions $n_s, n_i \ll n_{2\omega}$ throughout the waveguide, the OPA stage is undepleted. Substituting $A_{s,i}(z)=e^{-\alpha_\omega z}\tilde{A}_{s,i}(z)$ in \eqref{eq:fields} cancels the signal and idler loss terms, so $\tilde{A}_s$ and $\tilde{A}_i$ admit the standard two-mode parametric amplifier solution. The output photon numbers $n_{s,i}(L) = e^{-2\alpha_\omega L}\,|\tilde{A}_{s,i}(L)|^2$ are then
\begin{align}
    n_i(L) &= e^{-2\alpha_\omega L}\bigl(n_s(0) + 1\bigr)\sinh^{2}\bigl(\xi\bigr),    \label{eq:n_idler_int}\\
    n_s(L) &= e^{-2\alpha_\omega L}\bigl[n_s(0)\,\cosh^{2}\!\bigl(\xi\bigr) + \sinh^{2}\!\bigl(\xi\bigr)\bigr],\label{eq:n_signal_int}\\
\intertext{where}
    \xi &= \kappa_{\mathrm{SHG}}\!\int_0^{L}\! |A_{2\omega}(z)|\,dz, \label{eq:xi_def}
\end{align}
with the $+1$ inside the parentheses representing the vacuum seed responsible for spontaneous parametric fluorescence. The idler photon number at the detector is $n_i^{\mathrm{det}} = \eta_{\mathrm{out}}\,n_i(L)$.

In the lossless limit $\alpha_\omega = \alpha_{2\omega} = 0$, the Manley--Rowe relation $n_\omega(z) + 2\,n_{2\omega}(z) + n_s(z) + n_i(z) = n_\omega(0) + n_s(0)$ reduces, in the undepleted-OPA regime, to $n_\omega(z) + 2\,n_{2\omega}(z) \approx n_\omega(0)$, collapsing the SHG sub-system to a single equation with the analytic solution

\begin{align}
    A_{2\omega}(z) &= \sqrt{n_\omega(0)/2}\tanh\left(\kappa_{\mathrm{SHG}}\sqrt{n_\omega(0)/2}z\right),
    \label{eq:pump-depletion}\\
    \xi &= \ln[\cosh\left(\kappa_{\mathrm{SHG}} L\sqrt{n_\omega(0)/2}\right)].
    \label{eq:xi_lossless}
\end{align}
Outside the lossless limit no closed form is available and we integrate the SHG sub-system numerically with independently measured propagation losses as fixed inputs, substituting the resulting $A_{2\omega}(z)$ into~\eqref{eq:n_idler_int}. 

\subsection{Derivation of MAP estimator}\label{app:estimator}
The estimator infers $(\eta_1,\,\eta_2,\,\kappa_{\mathrm{SHG}})$ from a
power sweep of bidirectional measurements in two stages. With forward
($F$) pumping the left facet is the input ($\eta_{\mathrm{in}}=\eta_1$,
$\eta_{\mathrm{out}}=\eta_2$); under backward ($B$) pumping the roles
swap. We work in the seeded regime $\eta_{\mathrm{in}} n_{s} \gg 1$ and
with near-degenerate signal and idler so that the OPA stage is undepleted and
the closed-form \eqref{eq:n_idler_int} applies directly, and wavelength dependent facet loss biases are minimal. \\

\noindent{\bf Ratio estimator.}
Evaluating the ratio of idler powers $R$ in the seeded low-gain limit
($\sinh\xi\approx\xi$, $\eta_{\mathrm{in}}\,n_s \gg 1$) and using off-chip equal
input pump and signal powers reduces it to
\begin{equation}
    \hat{R} =\sqrt{n_{i,F}/n_{i,B}}\approx\frac{\eta_1}{\eta_2},
    \label{eq:R_low_gain}
\end{equation}
where $n_{i,F}$ and $n_{i,B}$ are the forward and backward mean photon numbers \eqref{eq:n_idler_int}.
Combined with the prior mean of the throughput $\eta_1\eta_2$, this
fixes a starting point $(\eta_1^{(0)},\eta_2^{(0)})$ for the MAP optimization. \\
\noindent{\bf MAP estimator.}
The MAP estimator (Eq. \eqref{eq:map} in the main text) acts on the joint likelihood over the swept pump powers $\{P_k\}$ and the two propagation directions.
For pump index $k$ and direction $d\in\{F,B\}$ the 
detected idler mean photon numbers are
\begin{align}
    \bar{n}_{F,k} &= \eta_1\eta_2e^{-2\alpha_\omega L}n_s
        \sinh^{2}\left(\xi(\eta_1,P_k;\kappa_{\mathrm{SHG}})\right),
    \\
    \bar{n}_{B,k}&= \eta_1\eta_2e^{-2\alpha_\omega L}n_s
        \sinh^{2}\left(\xi(\eta_2,P_k;\kappa_{\mathrm{SHG}})\right),
\end{align}
following from $n_i^{\mathrm{det}} = \eta_{\mathrm{out}}\,n_i(L)$ and
\eqref{eq:n_idler_int}. Propagation loss is direction symmetric, so the prefactor
$\eta_1\eta_2e^{-2\alpha_\omega L}n_s$ is identical in both directions and the
saturation argument $\xi$ is the only carrier of asymmetry information.

When $n_s\gg1$, the photon number variances for forward and backward propagation are 
\begin{align}
    \sigma^2_{F,k} &= \bar{n}_{F,k}\left[2\eta_2e^{-2\alpha_\omega L}\sinh^2\!\left(\xi(\eta_1,P_k)\right)+1\right] \text{ and}\\
\sigma^2_{B,k} &= \bar{n}_{B,k}\left[2\eta_1e^{-2\alpha_\omega L}\sinh^2\!\left(\xi(\eta_2,P_k)\right)+1\right],
\end{align}
respectively. 
A direct photodetection measurement on a displaced thermal state with large displacement has an approximately Gaussian distribution, with added electronic noise $\sigma^2_e$. Thus, we model our measurement statistics for forward and backward passes as
\begin{align}
    n_{F,k}&\sim \mathcal{N}(\bar{n}_{F,k},\sigma^2_{F,k} + \sigma^2_e)\\
    n_{B,k}&\sim \mathcal{N}(\bar{n}_{B,k},\sigma^2_{B,k} + \sigma^2_e) \label{eq:bdmap-meas-stats}
\end{align}
and our likelihood function $f$ is defined over the joint distribution of these independent measurement outcomes.

An independent insertion-loss measurement constrains the end-to-end transmission $\eta_1\eta_2 e^{-2\alpha_\omega L}$. Normalizing by the independently measured propagation transmission yields the facet product $\eta_{\mathrm{tot}}\equiv\eta_1\eta_2$, which we encode as a Gaussian prior with relative uncertainty $\sigma_{\mathrm{rel}}$, 
\begin{align}
    p\left(\eta_1\eta_2\right)
    &=\mathcal{N}\!\bigl(\eta_{\mathrm{tot}},\,\sigma_\eta^{2}\bigr),
    \\
    \sigma_\eta &= \sigma_{\mathrm{rel}}\,\eta_{\mathrm{tot}}.
    \label{eq:prior}
\end{align}

We require no prior on $\kappa_{\mathrm{SHG}}$. While measurements from a single pump power cannot separate nonlinear efficiency from coupling efficiency, saturation curvature over a sweep of pump powers breaks this degeneracy, making all three parameters jointly identifiable, as validated through the Monte Carlo simulation in the following subsection. 

\subsection{Simulation performance}\label{subsec:bd_simulation_perf}
We characterize the estimators via a Monte Carlo simulation over a 2D grid of true asymmetry $\eta_1^{(\mathrm{dB})}-\eta_2^{(\mathrm{dB})}\in[-4,4]$~dB and SHG efficiency $\kappa_{\mathrm{SHG}}\in[0.5,15]$, spanning the linear regime through to the onset of saturation at fixed pump. The total intrinsic throughput $\eta_1\eta_2$ is held at $-8$~dB. For each cell we synthesize $400$ trials of forward and backward measurements at three relative pump powers $(0.5,1,2)\times P_{\mathrm{mid}}$ with $P_{\mathrm{mid}}=0.05$~W, doubling to $800$ trials along the row and column corresponding to the 1D slices in Fig.~\ref{fig:sim_performance}. Each idler reading is drawn from a Gaussian with mean and variance given by \eqref{eq:bdmap-meas-stats}, with electronic readout noise $\sigma_e\approx1.5\!\times10^{3}$ photon-equivalents per integration. The MAP estimator uses the same likelihood so that its scatter can be benchmarked directly against the Cram\'er-Rao bound (CRB) obtained by inverting the $3\times3$ Fisher information matrix in $(\ln\eta_1,\ln\eta_2,\ln\kappa_{\mathrm{SHG}})$, including both mean and variance derivatives.

Figure~\ref{fig:sim_performance} summarizes the result. The MAP estimator is unbiased across the entire $(\kappa_{\mathrm{SHG}},\eta_1-\eta_2)$ grid, and its standard deviation tracks the CRB across the saturation-onset portion of the sweep ($\kappa_{\mathrm{SHG}}\gtrsim2$). The ratio estimator agrees with the MAP fit at small $\kappa_{\mathrm{SHG}}$ but acquires a bias of a few tenths of a dB at the highest $\kappa_{\mathrm{SHG}}$ once the seeded low-gain approximation behind \eqref{eq:R_low_gain} breaks down, and its scatter sits roughly $1.5{-}2\times$ above the CRB throughout. Over the full range the MAP estimator is therefore simultaneously unbiased and efficient.

\begin{figure*}
    \centering
    \includegraphics[width=\linewidth]{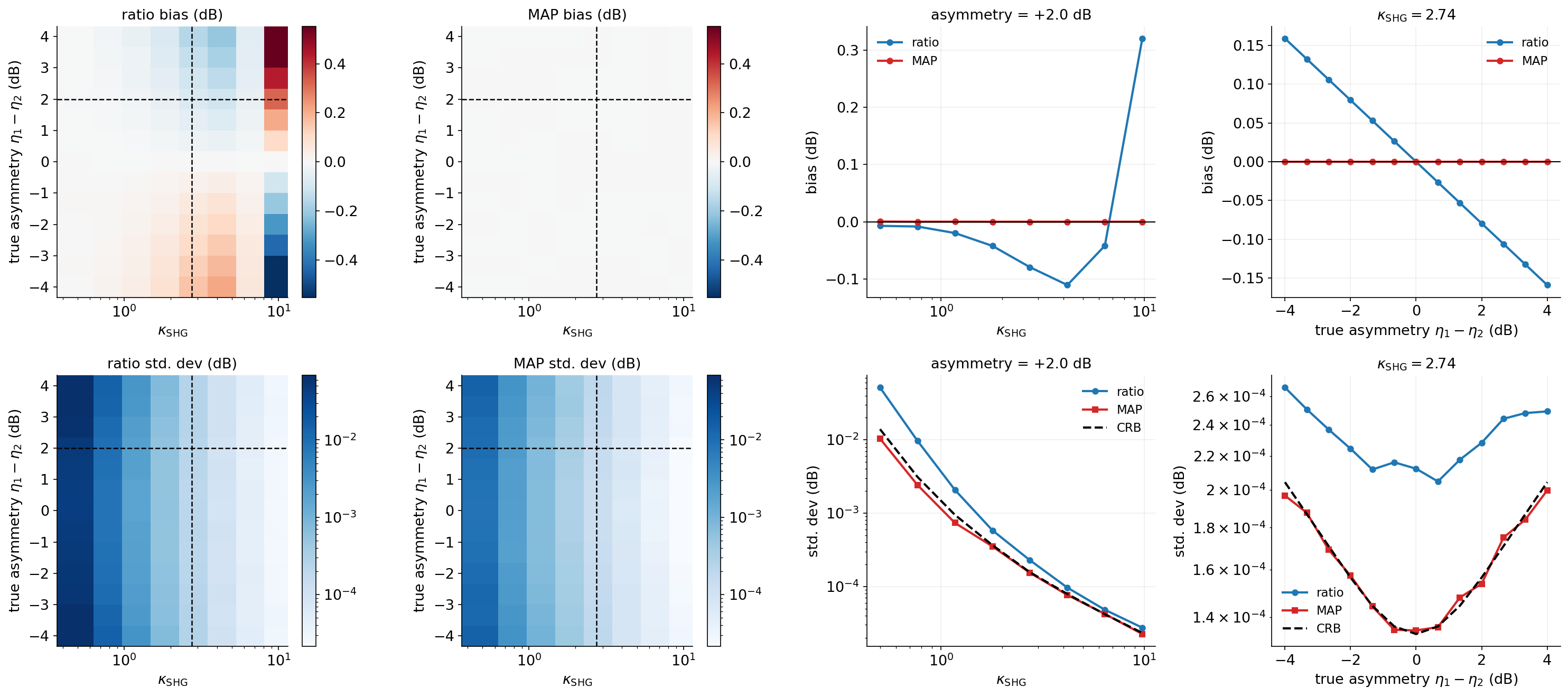}
    \caption{Estimator performance from Monte Carlo simulation across a 2D
    sweep of true asymmetry $\eta_1^{(\mathrm{dB})}-\eta_2^{(\mathrm{dB})}$
    and SHG efficiency $\kappa_{\mathrm{SHG}}\in[0.5,15]$.
    \emph{Left two columns:} bias (top) and standard deviation (bottom)
    of the asymmetry estimate as heatmaps for the ratio and MAP
    estimators; dashed lines mark the 1D slices shown in the right
    panels. \emph{Right two columns:} 1D bias and std.\ dev.\ slices
    against $\kappa_{\mathrm{SHG}}$ at fixed true asymmetry $+2$~dB and
    against true asymmetry at fixed $\kappa_{\mathrm{SHG}}\!=\!3.16$,
    with the CRB overlaid on the std.\ panels. The MAP fit is unbiased
    and hugs the CRB across the swept region; the ratio seed begins to
    deviate once the measurement leaves the seeded low-gain regime.}
    \label{fig:sim_performance}
\end{figure*}
\pagebreak
\section{Fabrication and Methods} \label{app: fab_pole}
\begin{figure}
    \centering
    \includegraphics[width=0.9\linewidth]{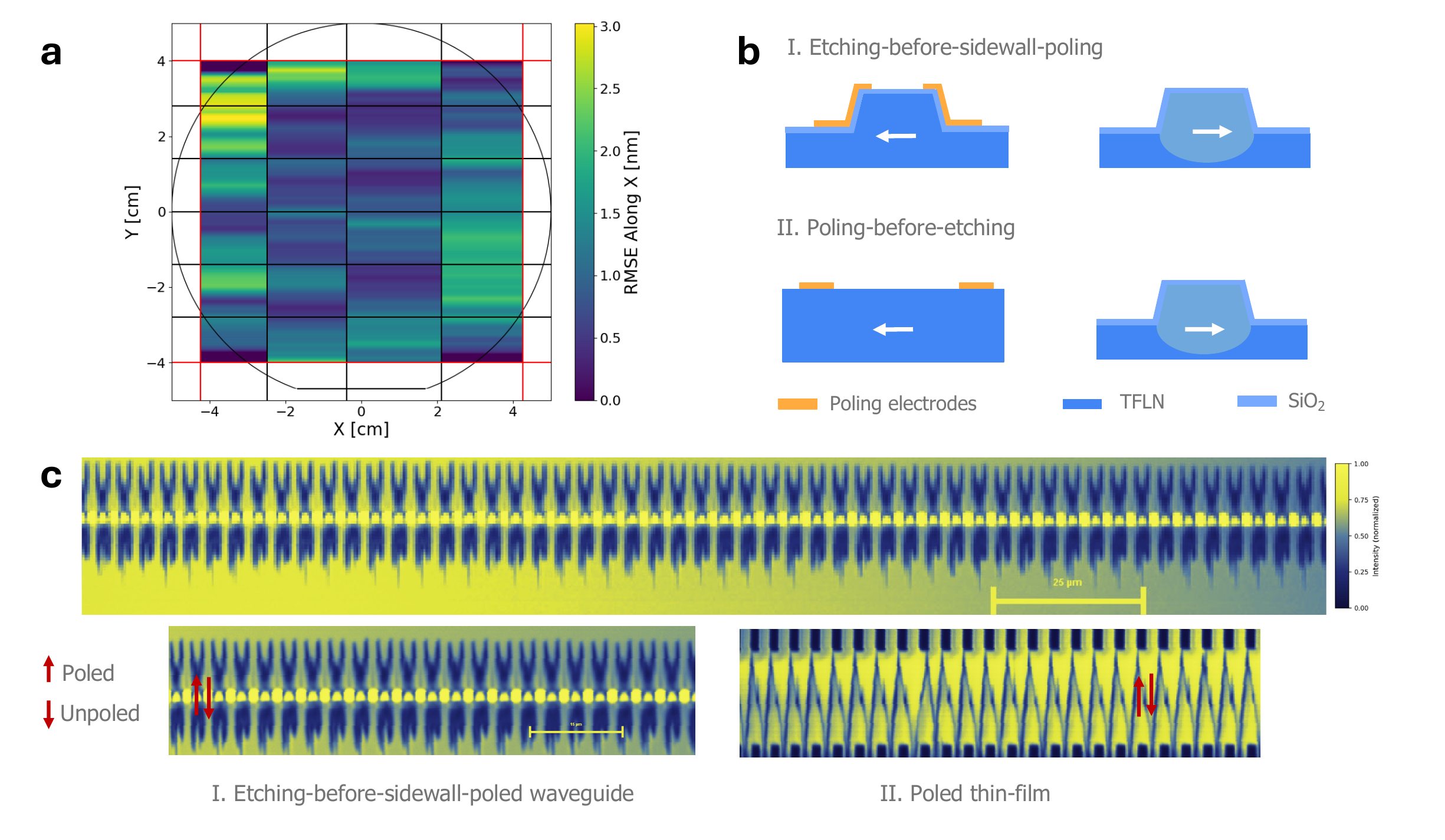}
    \caption{Low-loss sidewall-poling fabrication and propagation-loss characterization. (a) A full-wafer ellipsometry thickness map is used to define the chip dicing layout. The root-mean-square error (RMSE) of the film thickness along the optical propagation direction identifies regions with minimal longitudinal variation in thickness. (b) Comparison of the two fabrication strategies investigated in this work. (I) \textit{Etching-before-sidewall-poling}: waveguides are first defined by dry etching and subsequently periodically poled using sidewall electrodes deposited on a conformal 100-nm-thick SiO$_2$ layer. (II) \textit{Poling-before-etching}: a poled thin film is subsequently patterned into waveguides, followed by SiO$_2$ cladding deposition. (c) Second-harmonic microscopy images of periodically poled devices. Top, centimeter-scale sidewall-poled waveguide exhibiting uniform periodic domain inversion. Bottom, zoomed-in views of (I) an etched-before-poled waveguide and (II) a poled thin film prior to etching, both showing a nearly $50{:}50$ duty cycle and uniform domain contrast.}
    \label{fig:supp_mat_1}
\end{figure}

\subsection{Thin-film characterization} \label{app:tf-char}
Our fabrication process begins with wafer-scale ellipsometric mapping to quantify thickness nonuniformity and define an optimal dicing layout, as illustrated in Fig~\ref{fig:supp_mat_1}a. The process begins with an RMSE-minimizing algorithm that interpolates the wafer's height data onto a regular grid of $z$-values, indexed by $i = 1, \dots, n$ in the $x$-direction and $j = 1, \dots, m$ in the $y$-direction, giving grid points $z_{ij}$. The algorithm is then seeded with a set of $x$ indices $i_{0...d}$ for the initial boundaries, with $i_0=0$ and $i_d=n$. Since propagation predominantly occurs in the $x$ direction, $y$ boundaries have minimal impact on propagation RMSE and can be set arbitrarily afterwards. Once $i_{0...d}$ are given, the algorithm iteratively shifts $i_{1...d-1}$ in either direction to minimize the reward function $J=\sum_{j=1}^m\sum_{k=1}^{d}\sum_{i=i_{k-1}+1}^{i_{k}}w_{ij}(z_{ij}-\mu_{kj})^2$, with $\mu_{kj}=\sum_{i=i_{k-1}+1}^{i_{k}}w_{ij}z_{ij}/\sum_{i=i_{k-1}+1}^{i_{k}}w_{ij}$ being the mean $z$ value of an $x$ line within a die, and $w_{ij}$ being a weighting term for discarding points outside of the wafer. The entire function can be thought of as minimizing the variance length product within each die along the $x$ direction. A corrective term is also added to dissuade die lengths above or below cutoff values. We then select chip regions that minimize the root-mean-square error (RMSE) in film thickness along the optical propagation direction. Periodic poling is then performed exclusively within these low-RMSE regions to suppress phase mismatch arising from thickness variations. In addition, we employ adaptive poling, in which the poling period is locally adjusted to compensate for post-etching residual thickness variations, ensuring high-quality QPM across the entire nonlinear interaction length. Individual chips were diced from these regions and annealed at $520^{\circ}\mathrm{C}$ for approximately 2~h in an inert atmosphere prior to fabrication.  
\label{app: Thin-film charecterization}

\subsection{Fabrication and Poling}
\begin{figure}
    \centering
        \includegraphics[width=0.85\linewidth]{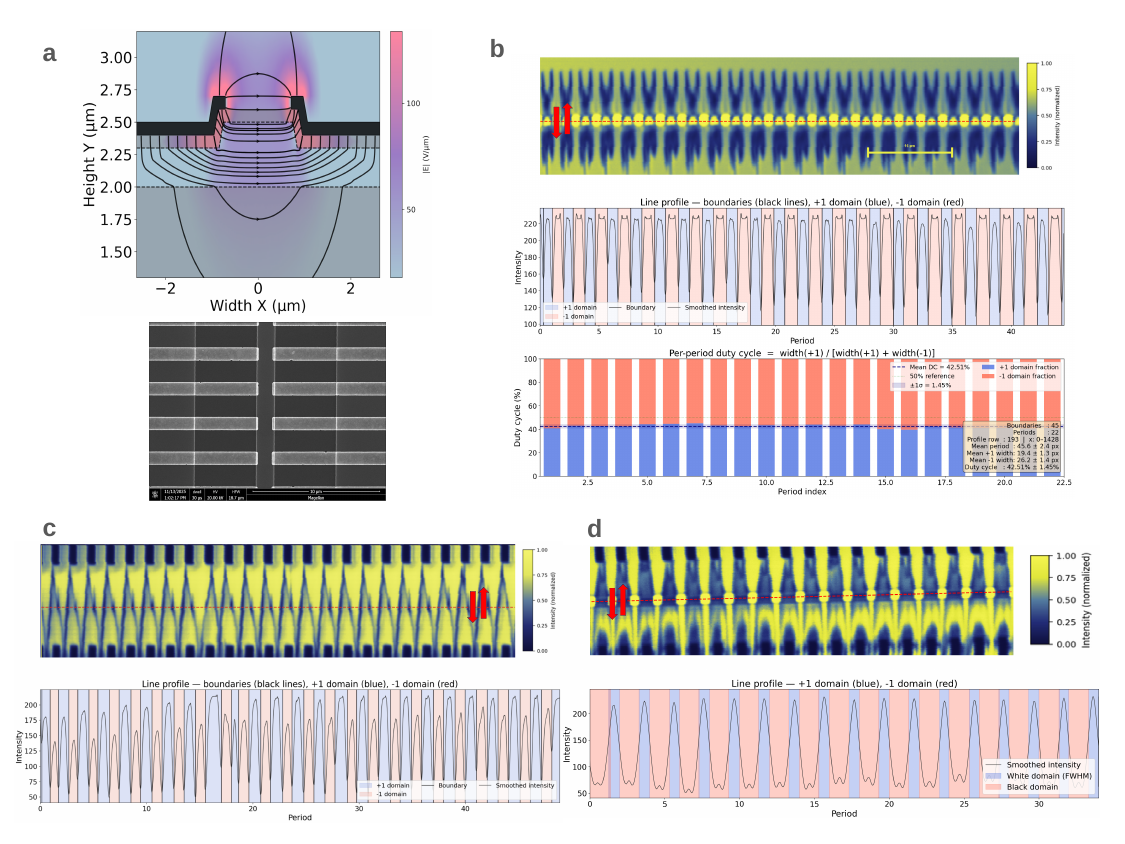}
    \caption{Duty-cycle and domain-inversion analysis for sidewall- and in-plane-poled devices. (a) Top: Simulated electric-field distribution across the waveguide cross-section for a $150~\mathrm{V}$ potential difference applied to the sidewall electrodes, showing strong field penetration throughout both the ridge and slab regions; bottom: scanning electron micrograph of the fabricated sidewall-electrode geometry.  (b--d) Second-harmonic microscopy images and corresponding cross-sectional domain profiles for \textbf{b,} sidewall-poled, \textbf{c,} in-plane poled-before-etching, and \textbf{d,} in-plane poled-after-etching devices, highlighting differences in duty-cycle uniformity and domain-inversion depth among the three poling configurations.}
    \label{fig:supp_mat_fig3}
\end{figure}

We investigate two distinct fabrication strategies, illustrated in Fig.~\ref{fig:supp_mat_1}b. The devices are fabricated on 500-nm-thick MgO-doped x-cut TFLN wafers (NanoLN), with a waveguide top width of $\sim1.2~\mu\mathrm{m}$, an etch depth of $\sim 200~\mathrm{nm}$, and a sidewall angle of 65 degree.  In the first, an \textit{etching-before-poling} process, near-single-mode waveguides designed for operation at $1550~\mathrm{nm}$ are first defined by dry etching. A detailed investigation of the etching and post-processing procedures will be presented in a separate study. Following waveguide fabrication, ellipsometry and profilometry measurements are performed to accurately determine the etched waveguide geometry. These measured dimensions are subsequently used to calculate poling periods, ensuring robust phase matching despite fabrication-induced dimensional variations. Periodic domain inversion is then performed using sidewall electrodes deposited on a conformal $\sim 100$-nm-thick SiO$_2$ layer, as illustrated in Fig.~\ref{fig:supp_mat_1}b(I). After poling, the electrodes are removed, and the devices undergo standard cleaning procedures prior to optical characterization.\\
In contrast, the \textit{poling-before-etching} approach involves periodic poling of the thin film prior to waveguide etching (Fig.~\ref{fig:supp_mat_1}b (II)). The poled film is subsequently etched and clad with a thin SiO$_2$ layer. Although this approach offers fabrication simplicity, including relaxed alignment constraints and a more robust lift-off process, it is inherently less tolerant to fabrication variations introduced during etching. As a result, deviations in the final waveguide geometry can lead to phase mismatches, making efficient nonlinear interactions at the targeted wavelengths challenging. We assess the domain quality using second-harmonic (SH) microscopy~\cite{reitzig2021seeing}. Figure~\ref{fig:supp_mat_1}c (top) shows the periodically poled \textit{etching-before-sidewall-poling} waveguide, exhibiting uniform domain inversion with a nearly $50{:}50$ duty cycle over centimeter-scale lengths. The bottom panels provide zoomed-in views for further comparison, where the downward and upward red arrows denote the unpoled and inverted domains, respectively. The left corresponds to the sidewall-poled waveguide, while the right panel shows the poled thin film prior to etching. In both cases, the central region of the poled structure exhibits a highly uniform domain pattern with a near-$50{:}50$ duty cycle and complete domain inversion, as evidenced by the nearly identical SH intensity from the poled and unpoled domains. 

In Fig.~\ref{fig:supp_mat_fig3}, we present our electrostatic simulations and duty cycle analysis.  Figure~\ref{fig:supp_mat_fig3}a shows that a $150~\mathrm{V}$ bias generates an electric field of approximately $50~\mathrm{kV/mm}$ across both the ridge and slab regions of the waveguide, exceeding the coercive field required for domain inversion in TFLN ($\sim30~\mathrm{kV/mm}$)~\cite{Nagy19}. Sidewall poling (Fig~\ref {fig:supp_mat_fig3}-b) gives a mean duty cycle of $42.5\% \pm 1.5\%$. Sidewall electrodes act as a parallel-plate capacitor with the waveguide between them, producing a confined, uniform electric field across both the slab and ridge regions, thereby improving poling depth and uniformity. \\
Further improvements may be achieved by optimizing the cladding thickness and electrode geometry to enhance the uniformity and penetration depth of the poling field. Figure~\ref{fig:supp_mat_fig3}c shows the poled thin film prior to waveguide etching, exhibiting an average duty cycle of $55.3\% \pm 5.8\%$ and appreciable spatial nonuniformity. A more pronounced limitation arises for the \textit{etching-before-poling} approach using conventional in-plane electrodes positioned away from the waveguide. In this geometry, the electric field is highly nonuniform across the waveguide cross-section and progressively weakens away from the positive electrode, limiting poling depth. As shown in Fig.~\ref{fig:supp_mat_fig3}d, domain inversion is largely confined to the slab region, leaving a substantial fraction of the waveguide ridge uninverted. This partial domain inversion is evident from the reduced SH-microscopy contrast of the inverted domains (black), in marked contrast to the sidewall-poled device in Fig.~\ref{fig:supp_mat_fig3}b, which exhibits substantially more uniform domain inversion across the waveguide cross-section~\cite{Sheltonetal}.

We pattern waveguides using 100-kV electron-beam lithography with a hydrogen silsesquioxane (HSQ) resist. The patterned HSQ served as a hard mask for the dry etching of the lithium niobate layer. The etching process was optimized to achieve an etch rate of approximately $22~\mathrm{nm/min}$ with a target etch depth of $200~\mathrm{nm}$. Following etching, the samples were cleaned in a hot RCA-1 solution ($55$--$60^{\circ}\mathrm{C}$) for 30~min to remove etch redeposition. A $\sim100~\mathrm{nm}$-thick SiO$_2$ cladding layer was subsequently deposited by plasma-enhanced chemical vapor deposition (PECVD), followed by annealing at $300^{\circ}\mathrm{C}$ for 2~hour to improve the oxide quality. Sidewall poling electrodes were then patterned in ZEP520A resist and fabricated by lift-off. Lift-off was performed in N-methyl-2-pyrrolidone (NMP), leaving the final electrode geometry as represented by the SEM image in Fig.~\ref{fig:supp_mat_fig3}(a). Although the SEM image corresponds to the device characterized in Appendix~\ref{app:sidewall-added-loss}, which was fabricated using ZEP rather than HSQ, the electrodes are representative of the sidewall electrodes used on the low-loss waveguides reported in this work. We then used a layer of photoresist to electrically insulate the surface and suppress breakdown during high-voltage sidewall poling. Poling was performed at high temperatures to enhance ferroelectric domain mobility~\cite{Doshietal}. Domain inversion was achieved by applying voltage pulses of $150$--$160~\mathrm{V}$ with millisecond durations. 

\label{app: Fabrication and Poling}
\label{app:fab_methods}
\vspace{-4mm}
\section{Characterization of poling-induced loss}\label{app:sidewall-added-loss}

\begin{figure}
    \centering
    \includegraphics[width=0.88\linewidth]{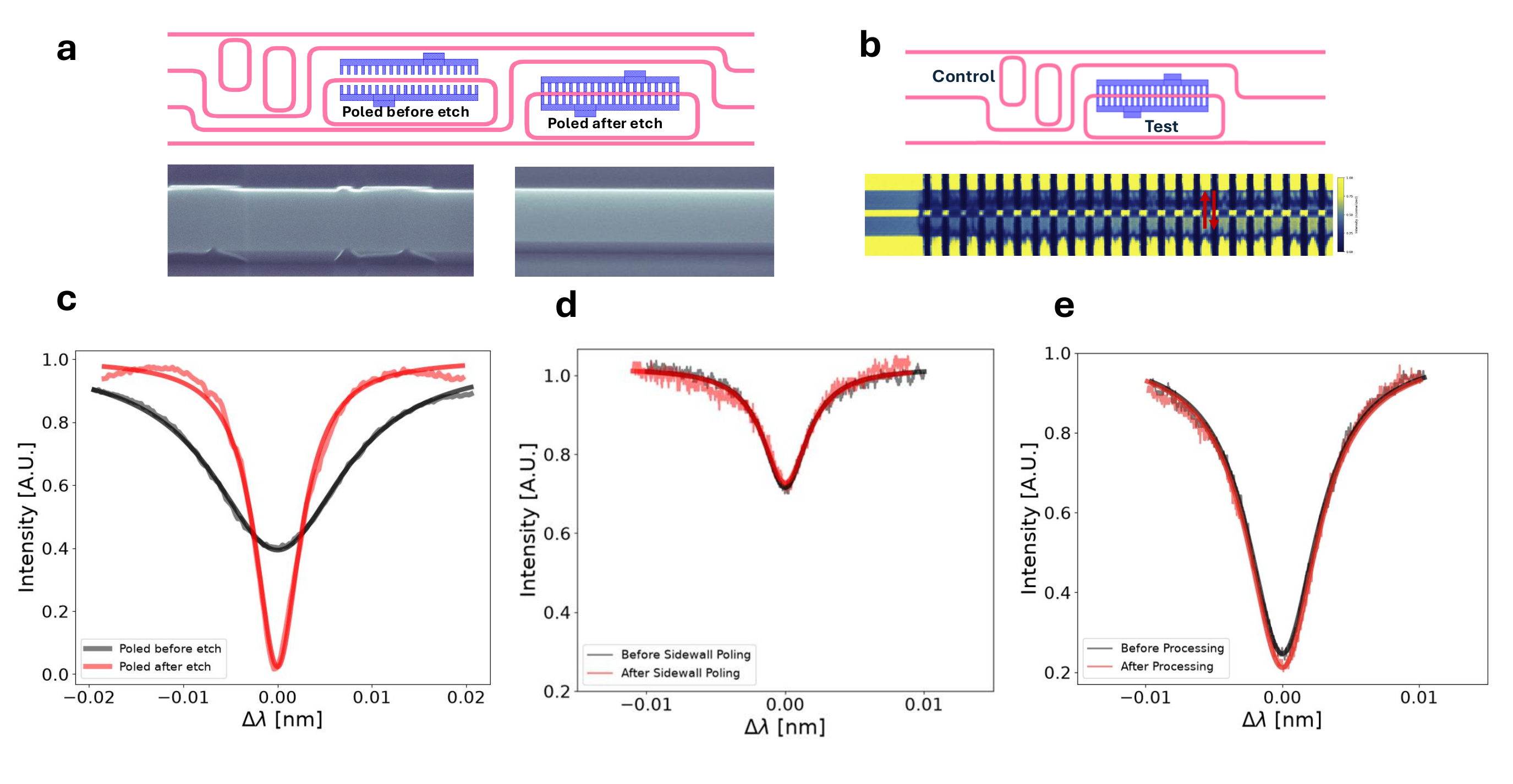}
    \caption{Comparison of propagation losses for the two fabrication strategies. (a) Top, resonator layout used for loss extraction. Bottom, scanning electron microscope images of a poled-before-etched waveguide (left) and an unpoled reference waveguide (right). (b) Top: Resonator layout used for before-and-after characterization. Bottom: Second-harmonic microscopy image of the intentionally over-poled resonator used to characterize the side-poling-induced loss. 
    (c) Normalized resonance spectra comparing the etch-before-sidewall-poled (red) and poled-before-etched (black) waveguides. (d) Resonance spectra of the `Test' resonator from (b) before and after sidewall poling show nearly identical linewidths and extinction ratios, confirming negligible added propagation loss. (e) The `Control' resonator shows the same behavior, confirming the poling process, metal removal, and cleaning steps did not alter propagation loss.}
    \label{fig:supp_mat_2}
\end{figure}

Next, we characterize propagation loss in periodically poled waveguides fabricated with each approach using a resonator-based technique, in which the periodically poled waveguides are embedded within resonators, as per layouts in Fig.~\ref{fig:supp_mat_2}a and b.  Figure~\ref{fig:supp_mat_2}c compares the normalized transmission spectra of resonators fabricated using the two poling strategies. We extract the propagation loss and quality factors from Lorentzian fits to resonator responses~\cite{McKinnon:09,younesi2025fabrication}. The \textit{etching-before-poling} device (red trace) exhibits a high-contrast resonance with an intrinsic quality factor of $Q_i \approx 700~\mathrm{k}$, consistent with low propagation loss in the periodically poled waveguide. 
In contrast, the \textit{poling-before-etching} resonator (black trace) shows a broadened resonance with reduced extinction and $Q_i \approx 119~\mathrm{k}$. We attribute this degradation to enhanced scattering from increased sidewall roughness arising from domain-dependent differential etching in pre-poled regions, which introduces roughness and structural irregularities along the waveguide propagation (see Fig.~\ref{fig:supp_mat_2}a, bottom left). The poled region makes up $26.6\%$ and $33.6\%$ of the total resonator 
length for the sidewall and poled-before-etched resonators, respectively.

Finally, we quantify the propagation loss introduced by the sidewall-poling process by characterizing a resonator before and after poling on a separate chip. To establish a conservative upper bound, the device is intentionally over-poled until the inverted domains exhibit lateral spreading beyond the target duty cycle. Lorentzian fits to the resonance spectra [Fig.~\ref{fig:supp_mat_2}d] yield intrinsic quality factors of $Q_i = 490.7 \pm 81.5~\mathrm{k}$ before poling and $Q_i = 462.3 \pm 70.4~\mathrm{k}$ after poling for the $3344~\mu\mathrm{m}$-long resonator, obtained by averaging $\sim20$ resonances spanning over 10 nm, respectively. The corresponding cavity round-trip losses are $0.28 \pm 0.05~\mathrm{dB}$ before poling and $0.29 \pm 0.04~\mathrm{dB}$ after poling and metal removal, respectively.
This difference is well within the experimental uncertainty, indicating that sidewall poling introduces negligible additional propagation loss. As a further control, we characterized resonators adjacent to the poled devices that underwent the identical fabrication sequence. These control resonators exhibit nearly identical resonance line widths and extinction ratios before and after processing, confirming that the fabrication sequence itself does not measurably degrade waveguide quality. Note that a comprehensive investigation of poling-induced loss in the sidewall-poling approach using high-$Q$ resonators will be presented in future work.

\end{document}